\documentclass[envcountsame]{llncs}

\newif\ifllncs
\llncstrue 

\newif\ifieee      \ieeefalse 

\newif\ifblind     \blindfalse 

\usepackage{latexsym}
\usepackage{amsmath}
\usepackage{amsfonts} 
\usepackage{amssymb}

\ifllncs\usepackage{smallsubsub}\else\usepackage{amsthm}\fi

\ifieee

\renewcommand\paragraph[1]{\par\bigskip\noindent\textbf{#1}\hspace{1ex}}
\fi

\usepackage{hyperref} 

\usepackage{xy} 
\xyoption{all} 
\usepackage{wrapfig}

\newcounter{running}[section]
\newenvironment{renumerate}{\begin{enumerate}\setcounter{enumi}{\value{running}}}%
{\setcounter{running}{\value{enumi}}\end{enumerate}}

\def\squareforqed{\hbox{\rlap{$\sqcap$}$\sqcup$}}
\newcommand{\myqed}{\ifieee{}\else\hfill\ensuremath{\squareforqed}\fi}

\newcommand{\bnd}{\ensuremath{\mathcal{B}}}
\newcommand{\bndA}{\ensuremath{\mathcal{A}}}
\newcommand{\Channels}{\ensuremath{\mathcal{CH}}}
\newcommand{\Locations}{\ensuremath{\mathcal{LO}}}
\newcommand{\Data}{\ensuremath{\mathcal{D}}}
\newcommand{\Endpoints}{\ensuremath{\mathcal{EP}}}
\newcommand{\Frame}{\ensuremath{\mathcal{F}}}
\newcommand{\Events}{\ensuremath{\mathcal{E}}}

\newcommand{\qdot}{\,\mathbf{.}\,}
\newcommand{\seq}[1]{ \langle {#1}  \rangle } 
\newcommand{\lts}[1]{\ensuremath{\mathsf{lts}(#1)}}
\newcommand{\restricted}{\mathbin{\mid\!\hspace{-1.17pt}\grave{}}}
\newcommand{\restrict}[2]{#1 \restricted\, #2}

\ifllncs
\spnewtheorem{eg}[theorem]{Example}{\bfseries}{\rmfamily}
\spnewtheorem{rmk}[theorem]{Remark}{\bfseries}{\rmfamily}
\else
\theoremstyle{definition} 
\newtheorem{theorem}{Theorem}
\newtheorem{definition}[theorem]{Definition}
\newtheorem{lemma}[theorem]{Lemma}
\newtheorem{corollary}[theorem]{Corollary}

\newtheorem{eg}[theorem]{Example}

\fi

\newenvironment{repeatthm}[2]{\medskip\par\noindent\textbf{#1\ 
    \ref{#2}.}  \em}{\medskip\par\noindent}

\newif\ifcites
\citesfalse   

\newcommand{\newop}[2]{\newcommand{#1}{\ensuremath{\mathsf{#2}}}}
\newop{\restr}{Rstr}
\newop{\exec}{Exc}
\newop{\ES}{ES}
\newop{\entry}{entry}
\newop{\exit}{exit}
\newop{\states}{states}
\newop{\initial}{initial}
\newop{\labels}{labels}
\newop{\accept}{recv}
\newop{\deliver}{xmit}
\newop{\mpair}{mpair}
\newop{\acceptep}{acc\_ep}
\newop{\deliverep}{del\_ep}
\newop{\createch}{create}
\newop{\pends}{ends}
\newop{\ipends}{init\_ends}
\newop{\LEFT}{left}
\newop{\RIGHT}{right}
\newop{\MID}{mid}
\newcommand{\compat}{J}
\newop{\local}{local}
\newop{\Tag}{Tag}
\newop{\sender}{sender}
\newop{\recipt}{rcpt} 
\newop{\chan}{chan} 
\newop{\chans}{chans} 
\newop{\msg}{msg} 
\newop{\fresh}{fresh} 
\newop{\events}{ev} 
\newop{\proj}{proj} 
\newop{\loc}{loc}
\newop{\Rt}{Rt}
\newop{\Rg}{Rg}
\newop{\gr}{gr}
\newop{\un}{un}
\newop{\run}{run}
\newop{\runs}{runs}
\newop{\intf}{intf}
\newop{\Host}{Host}
\newop{\traces}{traces}
\newop{\dom}{dom}
\newop{\invisible}{invis}
\newop{\visible}{vis}
\newop{\inp}{input}
\newop{\ND}{ND}
\newop{\NI}{NI}
\newop{\lab}{lab}
\newop{\select}{select}

\newop{\lsrc}{src}
\newop{\lobs}{obs}
\newop{\lcut}{cut}
\newop{\inb}{inb}
\newop{\outb}{outb}
\newop{\Dir}{Dir}

\newop{\vis}{\mathsf{vis}}
\newop{\hid}{\mathsf{hid}}

\newcommand{\LEFTZ}{\LEFT_0}
\newcommand{\lcutz}{\lcut_0}

\newcommand{\cmpt}[3]{\ensuremath{\compat_{{#2}\triangleleft{#1}}(#3)}}
\newcommand{\lruns}[1]{\ensuremath{{#1}\mbox{-}\runs}}
\newcommand{\lobsequiv}{\ensuremath{\approx_{\lobs}}}
\newcommand{\fequiv}{\ensuremath{\approx^{p}}}    

\newcommand{\lrunsone}[1]{\ensuremath{{#1}\mbox{-}\runs_1}}
\newcommand{\lrunstwo}[1]{\ensuremath{{#1}\mbox{-}\runs_2}}
\newcommand{\cmptone}[3]{\ensuremath{\compat^1_{{#2}\triangleleft{#1}}(#3)}}
\newcommand{\cmpttwo}[3]{\ensuremath{\compat^2_{{#2}\triangleleft{#1}}(#3)}}

\def\pushright#1{{              
    \parfillskip=0pt            
    \widowpenalty=10000         
    \displaywidowpenalty=10000  
    \finalhyphendemerits=0      
    \leavevmode                 
    \unskip                     
    \nobreak                    
    \hfil                       
    \penalty50                  
    \hskip.2em                  
    \null                       
    \hfill                      
    {#1}                        
    \par}}                      

\newcommand{\defsymbol}{\ensuremath{\scriptstyle /// }}
\newcommand{\defend}{\pushright{\defsymbol}\penalty-700 \smallskip}

\title{A Cut Principle for Information Flow\thanks{Copyright
    \copyright\ 2015 The MITRE Corporation.  All rights reserved.  
  }}
\author{Joshua D. Guttman and Paul D. Rowe} 
\ifllncs\institute{The MITRE Corporation\\\{guttman,prowe\}@mitre.org}
\fi

\begin{document}
\maketitle

\ifllncs\else{\small
  \tableofcontents}
\fi

\pagestyle{plain}

\begin{abstract}
  We view a distributed system as a graph of active locations with
  unidirectional channels between them, through which they pass
  messages.  In this context, the graph structure of a system
  constrains the propagation of information through it.  

  Suppose a set of channels is a cut set between an information source
  and a potential sink.  We prove that, if there is no disclosure from
  the source to the cut set, then there can be no disclosure to the
  sink.  We introduce a new formalization of partial disclosure,
  called \emph{blur operators}, and show that the same cut property is
  preserved for disclosure to within a blur operator.  A related
  compositional principle ensures limited disclosure for a class of
  systems that differ only beyond the cut.  \hfill\textbf{\today}

\end{abstract}

\section{Introduction}
\label{sec:intro}

In this paper, we consider information flow in a true-concurrency,
distributed model.  Events in an execution may be only partially
ordered, and locations communicate via synchronous message-passing.
Each message traverses a channel.  The locations and channels form a
directed graph.

Evidently, the structure of this graph constrains the flow of
information.  Distant locations may have considerable information
about each other's actions, but only if the information in
intermediate regions accounts for this.  If a kind of information does
not traverse the boundary of some portion of the graph (a \emph{cut
  set}), then it can never be available beyond that.  We represent
these limits on disclosure, i.e.~kinds of information that do not
escape, using \emph{blur operators}.  A blur operator returns a set of
behaviors local to the information source; these should be
indistinguishable to the observer.  When disclosure from a source to a
cut set is limited to within a blur operator, then disclosure to a
more distant region is limited to within the same blur operator (see
Thm.~\ref{thm:cut:II}, the \emph{cut-blur} principle).

Thus, disclosure to within a blur operator is similar to other forms
of \emph{what}-dimension
declassification~\cite{sabelfeld2009declassification}.  Blur operators
are our formalization of the semantic content of limited disclosures.
The cut-blur principle combines this with a \emph{where}-dimension
perspective.  It gives a criterion that localizes those disclosure
limits within a system architecture.

A related result, Thm.~\ref{cor:compose}, supports
\emph{compositional} security.  Consider any other system that differs
from a given one only in its structure beyond the cut.  That system
will preserve the flow limitations of the first, assuming that it has
the same local behaviors as the first in the cut set.  We illustrate
this (Examples~\ref{eg:compositional}--\ref{eg:voting:compositional})
to show that secrecy and anonymity properties of a firewall and a
voting system are preserved under some environmental changes.  Flow
properties of a simple system remain true for more complex systems, if
the latter do not distort the behavior of the simple system.


Our model is intended for many types of systems, including networks,
software architectures, virtualized systems, and distributed protocols
such as voting systems.  Our (somewhat informal) examples are meant to
be suggestive.  Network examples, which involve little local state,
are easy to describe, and rely heavily on the directed graph
structure.  Blur operators pay off by allowing many of their security
goals now to be regarded as information-flow properties.  Voting
systems offer an interesting notion of limited disclosure, since they
must disclose the result but not the choices of the individual voters.
Their granularity encourages composition, since votes are aggregated
from multiple precincts (Example~\ref{eg:voting:compositional}).

\paragraph{Motivation.}  A treatment of information flow---relying on
the graph structure of distributed systems---facilitates compositional
security design and analysis.

Many systems have a natural graph structure, which is determined early
in the design process.  Some are distributed systems where the
components are on separate platforms, and the communication patterns
are a key part of the specifications.  In other cases, the components
may be software, such as processes or virtual machines, and the
security architecture is largely concerned with their communication
patterns.  The designers may want to validate that these communication
patterns support the information flow goals of the design early in the
life cycle.  Thm.~\ref{cor:compose} justifies the designers in
concluding that a set of eventual systems all satisfy these security
goals, when those systems all agree on ``the part that really
matters.''


\paragraph{Contributions of this paper.}  Our main result is the
cut-blur principle, Thm.~\ref{thm:cut:II}.  The definition of
\emph{blur operator} is a supplementary contribution.  We show that
any reasonable notion of partial disclosure satisfies the conditions
for a blur (Lemma~\ref{lemma:partition}); these conditions are simple
structural properties that shape the proof of the cut-blur principle.
Thm.~\ref{cor:compose} brings the ideas to a compositional form.

\paragraph{Structure of this paper.}  After discussing motivating
examples (Section~\ref{sec:eg}) and some related work
(Section~\ref{sec:related}), we introduce our systems, called
\emph{frames}, and their execution model in Section~\ref{sec:model}.
This is a static model, in the sense that the channels connecting
different locations do not change during execution.
Section~\ref{sec:nondisclosure} proves the cut-blur principle for the
simple case of no disclosure of information at all across the
boundary.

We introduce blur operators (Section~\ref{sec:blur}) to formalize
partial disclosure.
In Section~\ref{sec:cut:blur}, we show that the cut idea is also
applicable to blurs (Thms.~\ref{thm:cut:II},~\ref{cor:compose}).
%
%
%

Section~\ref{sec:blur:non:interference} provides rigorous results to
relate our model to the literature.  We end by indicating some future
directions.  Appendix~\ref{sec:lemmas} contains longer proofs, and
additional lemmas.


\section{Two Motivating Examples} 
\label{sec:eg}

We first propose two problems we view in terms of information flow.
One is about network filtering; the other concerns anonymity in
voting.  In each, we want to prove an information flow result once,
and then reuse it compositionally under variations that do not affect
the part that really matters.

\begin{eg}[Network filtering]
  \label{eg:net:intro}
  \begin{figure}[tb]
    \centering
    {\footnotesize $ \xymatrix@R=4mm@C=4mm{
        && \bullet\ i\ar@/^/[d]\ar@(ru,rd)[]  && \\
        && \squareforqed\ r_1\ar@/^/[u]\ar@/^/[d] && \\
        && \squareforqed\ r_2\ar@/^/[u]\ar@/^/[dl]\ar@/^/[dr] && \\
        & \bullet\ {n_1}\ar@/^/[ur]\ar@(ld,lu)[] && \bullet\
        {n_2}\ar@/^/[ul]\ar@(rd,ru)[] & } $}
    \caption{A Two-Router Firewall}
    \label{fig:firewall}
  \end{figure}
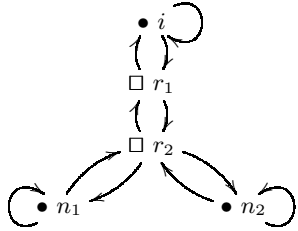
  Fig.~\ref{fig:firewall} shows a two-router firewall separating the
  public internet (node $i$) from two internal network regions
  $n_1,n_2$.  The firewall should ensure that any packet originating
  in the internal regions $n_1,n_2$ reaches $i$ only if it satisfies
  some property of its source and destination addresses, protocol, and
  port (etc.); we will call these packets \emph{exportable}.
  Likewise, any packet originating in $i$ reaches $n_1,n_2$ only if it
  satisfies a related property of its source and destination
  addresses, protocol, and port (etc.); we will call these packets
  \emph{importable}.

  These are information flow properties.  The policy provides
  \emph{confidentiality} for non-exportable packets within $n_1,n_2$,
  ensuring that they are not observable at $i$.  It provides a kind of
  \emph{integrity} protection for $n_1,n_2$ from non-importable
  packets from $i$, ensuring that $n_1,n_2$ cannot be damaged, or
  affected at all, if they are malicious.

  We assume here that packets are generated independently, so that
  (e.g.) no process on a host in $n_1,n_2$ generates exportable
  packets encoding confidential non-exportable packets it has sent or
  received.  If some process on a host is observing packets and coding
  their contents into packets to a different destination, this is a
  problem firewalls were not designed to solve, and security
  administrators worry about it separately.

  A firewall configuration enforcing a flow goal against the internet
  viewed as a single node $i$ should still succeed if $i$ has internal
  structure.  Similarly, the internal regions $n_1,n_2$ may vary
  without risk of security failure.  \defend
\end{eg}
We will return to this example several times to illustrate how we
formalize the system and specify its flow goals.
Example~\ref{eg:compositional} proves that some information flow goals
of Fig.~\ref{fig:firewall} remain true as the structure of $i,n_1,n_2$
varies.

\begin{eg}
  \label{eg:voting:intro} 
  As another key challenge, consider an electronic voting system such
  as ThreeBallot~\cite{rivest2007three}.
  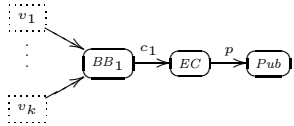
\begin{figure}[tb] \centering\tiny
    $$\xymatrix@R=.5mm@C=5mm{
      *+[F.]{v_1} \ar[rd] & & \\ \vdots &
*+[F-:<3pt>]{\mathit{BB}_1}\ar[r]^{c_1} &
*+[F-:<3pt>]{\mathit{EC}}\ar[r]^p & *+[F-:<3pt>]{\mathit{Pub}} \\
*+[F.]{v_k} \ar[ru] & & }
    $$
    \caption{A single precinct}
    \label{fig:precinct}
  \end{figure}
  Fig.~\ref{fig:precinct} shows the voters $v_1,\ldots,v_k$ of a
  single precinct, their ballot box ${\mathit{BB}_1}$, a channel
  delivering the results to the election commission ${\mathit{EC}}$,
  and then a public bulletin board ${\mathit{Pub}}$ that reports the
  results.

  The ballot box should provide voter anonymity:  neither
  ${\mathit{EC}}$ nor anyone observing the results ${\mathit{Pub}}$
  should be able to associate any particular vote with any particular
  voter $v_i$.  This also is an information flow goal.

  However, elections generally concern many precincts.
  Fig.~\ref{fig:election} contains $i$ precincts, all connected to the
  election commission $\mathit{EC}$.
  \begin{figure}[tb]
    \centering\tiny
    $$\xymatrix@R=.5mm@C=3mm{
      *+[F.]{v_1} \ar[rd] & & \\
      \vdots & *+[F-:<3pt>]{\mathit{BB}_1}\ar[rdd]^{c_1}
      &  \\
      *+[F.]{v_k} \ar[ru] & &  \\
      \vdots & & *+[F-:<3pt>]{\mathit{EC}}\ar[r]^p & *+[F-:<3pt>]{\mathit{Pub}} \\
      *+[F.]{w_1} \ar[rd] & & \\
      \vdots & *+[F-:<3pt>]{\mathit{BB}_i}\ar[ruu]_{c_i}
      & \\
      *+[F.]{w_j} \ar[ru] & &  }
    $$
    \caption{Multiple precincts report to $\mathit{EC}$}
    \label{fig:election}
  \end{figure}
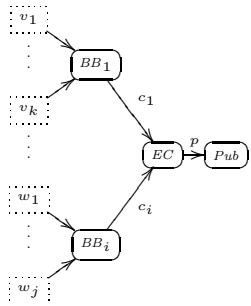
  Intuitively, a voter $v_n$ cannot lose their anonymity in the larger
  system:  ${\mathit{BB}_1}$ has already anonymized the votes in this
  first precinct.  Accumulating precinct summaries at ${\mathit{EC}}$
  cannot change the causal consequences of ${\mathit{BB}_1}$'s
  actions.  \defend
\end{eg}
We formalize the flow goals of this example in
Example~\ref{eg:voting}, and extend the results to
Fig.~\ref{fig:election} in Example~\ref{eg:voting:compositional}.

These simple examples illustrate the payoff from a compositional
approach to flow goals.  Conclusions about a firewall should be
insensitive to changes in the structure of the networks to which it is
attached.  An anonymity property achieved by a ballot mechanism should
be preserved as we collect votes from many precincts.  These are
situations where we want to design and reuse mechanisms, with a
criterion to ensure the mechanism is remains safe under changes in the
part that does not really matter.  Thm.~\ref{cor:compose} below is the
criterion we propose.


\section{Some related work}
\label{sec:related}  

We return to related work also in
Sections~\ref{sec:blur:non:interference}--\ref{sec:summary}.


\paragraph{Noninterference and nondeducibility.}  There is a massive
literature on information-flow security; Goguen and Meseguer were key
early contributors~\cite{GoguenMeseguer82}.  Sutherland introduced the
non-deducibility idea~\cite{Sutherland86} as a way to formalize lack
of information flow, which we have adopted in our ``non-disclosure''
(Def.~\ref{def:non:disclosure}).  Subsequent work has explored a wide
range of formalisms, including state
machines~\cite{rushby1992noninterference}; process algebras such as
CSP~\cite{Roscoe1995csp,RoscoeGoldsmith99,RyanSchneider2001} and
CCS~\cite{FocardiGorrieri97,focardi2001classification,bossi2004verifying};
and bespoke
formalisms~\cite{mccullough1988noninterference,mantel2000possibilistic}.

Irvine, Smith, and Volpano reinvigorated a language-based
approach~\cite{VolpanoEtAl96}, inherited from Denning and
Denning~\cite{denning1977certification}, in which systems are
programs.  Typing ensures that their behaviors satisfy
information-flow goals.  Sabelfeld and Myers~\cite{SabelfeldMyers03}
reviewed the early years of this approach.  Distributed execution has
been considered in language-based work,
e.g.~\cite{zdancewic2002secure,chong2007secure,bohannon2009reactive}.
Our work here is not specifically language-based, since the behaviors
of our locations are defined only as sets of traces.  If these
behaviors are specified as programs, however, our results may be
usefully combined with language-based methods.

\paragraph{Declassification.}  Declassification is a major concern for
us.  A blur operator (Def.~\ref{def:blur}) determines an upper bound
on what a system may declassify.  It may declassify information unless
its blur operators require those details to be blurred out.  Like
escape-hatches~\cite{sabelfeld2004model}, this is disclosure along the
\emph{what}-dimension, in the Sabelfeld-Sands
classification~\cite{sabelfeld2009declassification}.  The cut-blur
principle connects this \emph{what} declassification to \emph{where}
the processing responsible for the declassification will occur in a
system architecture.  In this regard, it combines a semantic view of
what information is declassified with an architectural view related to
intransitive
noninterference~\cite{rushby1992noninterference,van2007indeed}.
Balliu et al.~\cite{balliu2011epistemic} connect \emph{what},
\emph{where}, and \emph{when} declassification via epistemic logic,
although without a compositional method.  

McCamant and Ernst~\cite{mccamant2008quantitative} study quantitative
information flow when programs run.  A directed acyclic graph
representing information flow is generated dynamically by inspecting a
particular execution or set of executions.  The max-flow/min-cut
theorem then bounds flow through the program by what can traverse
minimal cuts.  Apparently, they do not ensure that every possible
execution will respect the bounds.  They do not provide a
compositional way to reuse flow conclusions like our main theorems.

\paragraph{Composability and refinement.}  McCullough first raised the
questions of non-determinism and composability of information-flow
properties~\cite{mccullough1987specifications,mccullough1988noninterference}.
This was a major focus of work through much of the period since,
persisting until
today~\cite{johnson1988security,zakinthinos1995composability,mantel2001preserving,mantel2002composition,rossi2009information,RafnssonSabelfeld14}.
Mantel, Sands, and Sudbrock~\cite{mantel2011assumptions} use a
rely/guarantee method for compositional reasoning about flow in the
context of imperative programs.
Roscoe~\cite{roscoe1994non,Roscoe1995csp} offers a definition based on
determinism, which is intrinsically composable.
Morgan's~\cite{morgan2006shadow} programming language treatment
clarifies the refinements that preserve security.  Our results do not
run afoul of the refinement paradox
either~\cite{jacob1992basic,morgan2006shadow}:  our theorems identify
the assumptions that ensure that blurs are preserved.

Van der Meyden~\cite{van2012architectural} provides an architectural
treatment designed to achieve preservation under refinement.  Our work
is distinguished from it in offering a new notion of composition,
illustrated in Examples~\ref{eg:net:intro}--\ref{eg:voting:intro}; in
focusing on declassification; and in applying uniformly to a range of
declassification policies, defined by the blur operators.

Van der Meyden's work with
Chong~\cite{chong2009deriving,ChongvanderMeyden14} is most closely
related to ours.  They consider ``architectures,'' i.e.~directed
graphs that express an intransitive noninterference style of
\emph{what}-dimension flow policy.  The nodes of an architecture are
security domains, intended to represent levels of information
sensitivity.  The authors define when a (monolithic) deterministic
state machine, whose transitions are annotated by domains,
\emph{complies} with an architecture.  The main result
in~\cite{chong2009deriving} is a cut-like epistemic property on the
architecture graph:  Roughly, any knowledge acquired by a recipient
about a source implies that the same knowledge is available at every
cut set in the architecture graph.

A primary contrast between this paper and~\cite{chong2009deriving} is
our distributed execution model.  We consider it a more localized link
to development, since components are likely to be designed,
implemented, and upgraded piecemeal.  Chong and van der Meyden focus
instead on the specification level, in which security specifications
(rather than system components) form the directed graph.  This is a
new and unfamiliar artifact needed before analysis can occur.  Their
epistemic logic allows nested occurrences of the \emph{knowledge}
modality $K_G$, or occurrences of $K_G$ in the hypothesis of an
implication.  However, this surplus extra expressiveness is not used
in their examples, which do not have nested $K_G$ operators, or
occurrences of $K_G$ in the hypothesis of an implication.  Indeed, our
clean proof methods suggest that our model may have the right degree
of generality, and be easy to understand, apply, and enrich.

Recently~\cite{ChongvanderMeyden14}, they label the arrows by
functions $f$, where $f$ filters information from its source, bounding
visibility to its target.  They have not re-established their cut-like
epistemic property in the richer model, however.

Van der Meyden and Chong's refinement
method~\cite{van2012architectural,ChongvanderMeyden14}
%
%
%
applies when
%
the refined system has a homomorphism \emph{onto} the less refined
one.  Although it may cover Example~\ref{eg:net:intro}, it does not
cover Example~\ref{eg:voting:intro}, where the refined system contains
genuinely new components and events.


\section{Frames and Executions}
\label{sec:model}

We represent systems by \emph{frames}.  Each frame is a directed
graph.  Each node, called a \emph{location}, is equipped with a set of
traces defining its possible local behaviors.  The arrows are called
\emph{channels}, and allow the synchronous transmission of a message
from the location at the arrow tail to the location at the arrow head.
Each message also carries some \emph{data}.  

\subsection{The Static Model}
\label{sec:model:static}

In this paper, we will be concerned with a static version of the
model, in which channel endpoints are never transmitted from one
location to another, although we will briefly describe a dynamic
alternative in Section~\ref{sec:summary:dynamic}, in which channels
may also be delivered over other channels.

\begin{definition}[Frame] 
  \label{def:frame} 
  A (static) \emph{frame} $\Frame$ has three disjoint domains:
  \begin{description}
    \item[Locations $\Locations$:] Each location $\ell\in\Locations$
    is equipped with a set of traces, $\traces(\ell)$ 
%
%
    and other information, further constrained below.
    \item[Channels $\Channels$:] Each channel $c\in\Channels$ is
    equipped with two endpoints, $\entry(c)$ and $\exit(c)$.  It is
    intended as a one-directional conduit of data values between the
    endpoints.  
    \item[Data values $\Data$:] Data values $v\in\mathcal{D}$
    may be delivered through channels.
  \end{description}
  We will write $\Endpoints$ for the set of channel endpoints, which
  we formalize as $\Endpoints=\{\entry,\exit\}\times\Channels$,
  although we will generally continue to type $\entry(c)$ and
  $\exit(c)$ to stand for $\seq{\entry,c}$ and $\seq{\exit,c}$.
  For each $\ell\in\Locations$, there is:
  \begin{renumerate}
    \item a set of endpoints $\pends({\ell})\subseteq\Endpoints$ such
    that
    \begin{enumerate}
      \item $\seq{e,c}\in\pends({\ell})$ and
      $\seq{e,c}\in\pends({\ell'})$ implies $\ell=\ell'$; and
      \item there is an $\ell$ such that $\entry(c)\in\pends({\ell})$
      iff \\
      there is an $\ell'$ such that $\exit(c)\in\pends({\ell'})$;
    \end{enumerate}
  \end{renumerate}
    When $\entry(c)\in\pends({\ell})$ we will write $\sender(c)=\ell$;
    when $\exit(c)\in\pends({\ell})$ we will write $\recipt(c)=\ell$.
    Thus, $\sender(c)$ is the location that can send messages on $c$,
    while $\recipt(c)$ is the recipient location that can receive
    them.  We write $\chans(\ell)$ for $\{c\colon\sender(c)=\ell$ or
    $\recipt(c)=\ell\}$.  

    We say that $\lambda$ is a \emph{label} for $\ell$ if
    $\lambda=(c,v)$ where $c\in\chans(\ell)$ and $v\in\Data$.  We
    divide labels into \emph{local}, \emph{transmission}, and
    \emph{reception} labels, where $c,v$ is:
      \begin{description}
        \item[local to $\ell$] if $\sender(c)=\ell=\recipt(c)$;
        \item[a transmission for $\ell$] if $\sender(c)=\ell\not=\recipt(c)$;
        \item[a reception for $\ell$] if
        $\sender(c)\not=\ell=\recipt(c)$.  
      \end{description}

  \begin{renumerate}
    \item a prefix-closed set $\traces(\ell)$, each trace being a
    finite or infinite sequence of labels $\lambda$.  \defend
%
%
  \end{renumerate}
%
%
\end{definition}
In this definition, we do not require that the local behaviors
$\traces(\ell)$ should be determined in any particular way.  They
could be specified by associating a program to each location, or a
term in a process algebra, or a labeled transition system, or a
mixture of these for the different locations.

Each $\Frame$ determines directed and undirected graphs, $\gr(\Frame)$
and $\un\gr(\Frame)$:
\begin{definition}
  If $\Frame$ is a frame, then the \emph{graph of} $\Frame$, written
  $\gr(\Frame)$, is the directed graph $(V,E)$ whose vertices
  $V$ are the locations $\Locations$, and such that there is an edge
  $(\ell_1,\ell_2)\in E$ iff, for some $c\in\Channels$,
  $\sender(c)=\ell_1$ and $\recipt(c)=\ell_2$.  

  The undirected graph $\un\gr(\Frame)$ has those vertices, and an
  undirected edge $(\ell_1,\ell_2)$ whenever either $(\ell_1,\ell_2)$
  or $(\ell_2,\ell_1)$ is in the edges of $\gr(\Frame)$.  \defend
\end{definition}

\subsection{Execution semantics}
\label{sec:model:execution:semantics}

We give an execution model for a frame via partially ordered sets of
events.  The key property is that the events that involve any single
location $\ell$ should be in $\traces(\ell)$.  Our semantics is
reminiscent of Mattern~\cite{Mattern1989}, although his model lacks
the underlying graph structure.  We require executions to be
well-founded, but none of the results in this paper depend on the
assumption.

\begin{definition}[Events; Executions]
  \label{def:events}
  Let $\Frame$ be a frame, and let $\Events$ be a structure
  $\seq{E,\chan,\msg}$.  The members of $E$ are \emph{events},
  equipped with the functions:
  \begin{description}
    \item[$\chan\colon E\rightarrow\Channels$] returns the
    channel of each event; and
    \item[$\msg\colon E\rightarrow\Data$] returns the message
    communicated in each event.
    
  \end{description}
  \begin{enumerate}
    \item $\bnd=(B,\preceq)$ is a \emph{system of events}, written
    $\bnd\in\ES(\Events)$, iff (i) $B\subseteq E$; (ii) $\preceq$ is a
    partial ordering on $B$; and (iii), for every $e_1\in B$,
    $\{e_0\in B\colon e_0\preceq e_1\}$ is
    finite.\label{cl:finite:predecessors}
  \end{enumerate}
  Hence, $\bnd$ is well-founded.  If $\bnd=(B,\preceq)$, we refer to
  $B$ as $\events(\bnd)$ and to $\preceq$ as $\preceq_\bnd$.  
  Now let $\bnd=(B,\preceq)\in\ES(\Events)$:
  \begin{enumerate}\setcounter{enumi}{1}
    \item Define $\proj(B,\ell)=$
    $$\{e\in B\colon \sender(\chan(e))=\ell\mbox{ or
    }\recipt(\chan(e))=\ell\}.$$
    \item $\bnd$ is an \emph{execution}, written
    $\bnd\in\exec(\Frame)$ iff, for every $\ell\in\Locations$,
    \begin{enumerate} 
      \item $\proj(B,\ell)$ is linearly ordered by $\preceq$,
      and\label{cl:exec:linear}
      \item $\proj(B,\ell)$ yields an element of
      $\traces(\ell)$\label{cl:exec:trace} when viewed as a sequence.
      \defend 
       \end{enumerate}
     \end{enumerate} 
%
%
\end{definition}
By Def.~\ref{def:events}, Clauses~\ref{cl:finite:predecessors}
and~\ref{cl:exec:linear}, $\proj(B,\ell)$ is a finite or
infinite sequence, ordered by $\preceq$.  Thus
Clause~\ref{cl:exec:trace} is well-defined.  We will often write
$\bndA,\bndA'$, etc., when $\bndA,\bndA'\in\exec(\Frame)$.  For a
frame $\Frame$, the choice between two different structures
$\Events_1,\Events_2$ makes little difference:  If
$\Events_1,\Events_2$ have the same cardinality, then to within
isomorphism they lead to the same systems of events and hence also
executions.  Thus, we will suppress the parameter $\Events$,
henceforth.

\begin{definition} Let $\bnd_1=(B_1,\preceq_1),
  \bnd_2=(B_2,\preceq_2)\in\ES(\Frame)$.
  \begin{enumerate}
    \item $\bnd_1$ is a \emph{substructure} of $\bnd_2$ iff
    $B_1\subseteq B_2$ and
    $\mathop\preceq_1=({\preceq_2}\cap{B_1\times B_1})$.
    \item $\bnd_1$ is an \emph{initial substructure} of $\bnd_2$ iff
    it is a substructure of the latter, and for all $y\in B_1$, if
    $x\preceq_2 y$, then $x\in B_1$.  \defend
  \end{enumerate}
\end{definition}

\begin{lemma}
  \label{lemma:events}
  \begin{enumerate}
    \item If $\bnd_1$ is a substructure of $\bnd_2\in\ES(\Frame)$,
    then $\bnd_1\in\ES(\Frame)$.
    \item If $\bnd_1$ is an initial substructure of
    $\bnd_2\in\exec(\Frame)$, then $\bnd_1\in\exec(\Frame)$.
    \item Being an execution is preserved under chains of initial
    substructures: Suppose that $\seq{\bnd_i}_{i\in\mathbb{N}}$ is a
    sequence where each $\bnd_i\in\exec(\Frame)$, such that $i\le j$
    implies $\bnd_i$ is an initial substructure of $\bnd_j$.  Then
    $(\bigcup_{i\in\mathbb{N}}{\bnd_i})\in\exec(\Frame)$.  \defend
  \end{enumerate}
\end{lemma}

\begin{eg}[Network with filtering]
  \label{eg:net:idea}\label{eg:net:details}
  To localize our descriptions of functionality, we expand the network
  of Fig~\ref{fig:firewall}; see Fig.~\ref{fig:net:1}.  Regions are
  displayed as $\bullet$; routers, as $\squareforqed$; and interfaces,
  as $\bigtriangleup$.  When a router has an interface onto a segment,
  a pair of locations---representing that interface as used in each
  direction---lie between this router and each peer
  router~\cite{GuttmanHerzog05}.  

  Let $\Dir=\{\inb,\outb\}$ represent the inbound direction and the
  outbound directions from routers, respectively.  Suppose $\Rt$ is a
  set of routers $r$, each with a set of interfaces $\intf(r)$, and a
  set of network regions $\Rg$ containing end hosts.

  Each member of $\Rt,\Rg$ is a \emph{location}.  Each
  interface-direction pair
  $(i,r)\in(\bigcup_{r\in\Rt}\intf(r))\times\Dir$ is also a location.
  The \emph{channels} are those shown.  Each interface has a pair of
  channels that allow datagrams to pass between the router and the
  interface, and between the interface and an adjacent entity.  We
  also include a self-loop channel at each network region $i,n_1,n_2$;
  it represents transmissions and receptions among the hosts and
  network infrastructure coalesced into the region.  Thus: 
  \begin{description}
    \item[$\Locations=$] $\Rt\cup\Rg\cup((\bigcup_{r\in\Rt}\intf(r))\times\Dir)$;
    \item[$\Channels=$]
    $\{(\ell_1,\ell_2)\in\Locations\times\Locations\colon \ell_1$
    delivers datagrams directly to $\ell_2\}$;
    \item[$\Data=$] the set of IP datagrams;
    \item[$\pends(\ell)=$]
    $\{\entry(\ell,\ell_2)\colon
    (\ell,\ell_2)\in\Channels\}\cup\{\exit(\ell_1,\ell)\colon $
    $ (\ell_1,\ell)\in\Channels\}$, \ifllncs \\ \fi for each
    $\ell\in\Locations$.
  \end{description}
  Behaviors are easily specified.  Each router $r\in\Rt$ receives
  packets from inbound interfaces, and chooses an outbound interface
  for each.  Its state is a set of received but not yet routed
  datagrams, and the sole initial state is $\emptyset$.  The
  transition relation, when receiving a datagram, adds it to this set.
  When transmitting a datagram $d$ in the current set, it removes $d$
  from the next state and selects an outbound channel as determined by
  the routing table.  For simplicity, the routing table is an
  unchanging part of determining the transition relation.

  A directed interface enforces filtering rules.  The state again
  consists of the set of received but not-yet-processed datagrams.
  The transition relation uses an unchanging filter function to
  determine, for each datagram, whether to discard it or retransmit
  it.

  If $n\in\Rg$ is a region, its state is the set of datagrams it has
  received and not yet retransmitted.  It can receive a datagram;
  transmit one from its state; or else initiate a new datagram.  If
  assumed to be well-configured, these all have source address in a
  given range of IP addresses.  Otherwise, they may be arbitrary.
  \defend
\end{eg}
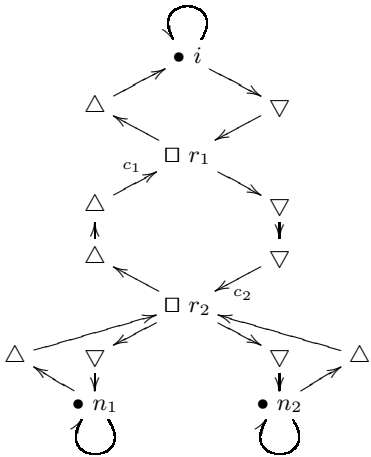
\begin{figure}[tb]  
  \centering 
%
  {\footnotesize $ \xymatrix@R=2mm@C=4mm{
      & & \bullet\ i\ar[dr]\ar@(ur,ul)[] & &\\
      &\bigtriangleup\ar[ur] & & \bigtriangledown\ar[dl] & \\
      & & \squareforqed\ r_1\ar[ul]\ar[dr] & & \\
      &\bigtriangleup\ar[ur]^{c_1} & & \bigtriangledown\ar[d] & \\
      &\bigtriangleup\ar[u] & & \bigtriangledown\ar[dl]^{c_2} & \\
      & & \squareforqed\ r_2\ar[ul]\ar[dl]\ar[dr] & & \\
      \bigtriangleup\ar[urr] & \bigtriangledown\ar[d] & &
      \bigtriangledown\ar[d] &  \bigtriangleup\ar[ull] \\
      & \bullet\ {n_1}\ar[ul]\ar@(dr,dl)[] && \bullet\
      {n_2}\ar[ur]\ar@(dr,dl)[] & } $} \vspace{4mm}
  \caption{Expanded representation of network from Fig.~\ref{fig:firewall}}
\label{fig:net:1}
\end{figure}

If the router is executing other sorts of processing, for instance
Network Address Translation or the IP Security Protocols, then the
behavior is slightly more
complex~\cite{adao2014mignis,GuttmanHerzog05}, but sharply localized.
Many other problems can be viewed as frames.  We have mentioned voting
schemes (Ex.~\ref{eg:voting:intro}).  An attestation architecture for
secure virtualized systems~\cite{CokerEtAl10} is, at one level, a set
of virtual machines communicating through one-directional channels.
%
%

\paragraph{Partially vs.~totally ordered executions.}
Def.~\ref{def:events} does not require the ordering $\preceq$ of
``occurring before'' to be total.  When events occur on different
channels, neither has to precede the other.  Thus, our executions need
not be sequential.

This has three advantages.  First, it is more inclusive, since
executions with total orders satisfy our definition as do those with
(properly) partial orders.  Indeed, the main claims of this paper
remain true when restricted to executions that are totally ordered.
Second, reasoning is simplified.  We do not need to interleave events
when combining two local executions to construct a global one, as
encapsulated in the proof of
Lemma~\ref{lemma:combine:executions:two:frames}.  Nor do we need to
``compact'' events, when splitting off a local execution, as we would
if we used a particular index set for sequences.  This was probably an
advantage to us in developing these results.  Third, the minimal
partial order is a reflection of causality, which can be used also to
reason about independence.  We expect this to be useful in future
work.

There is also a disadvantage:  unfamiliarity.  It requires some
caution.  Moreover, mechanized theorem provers have much better
support for induction over sequences than over well-founded orders.
This was an inconvenience when a colleague used PVS to formalize parts
of this work, eventually choosing to use totally ordered executions.
With that difference, Thms.~\ref{thm:cut:II} and~\ref{cor:compose}
have been confirmed in PVS, as have the basic properties of
Example~\ref{eg:voting:compositional}.


\section{Non-disclosure}
\label{sec:nondisclosure}

Following Sutherland~\cite{Sutherland86}, we think of information flow
in terms of \emph{deducibility} or \emph{disclosure}.  A participant
observes part of the system behavior, trying to draw conclusions about
a different part.  If his observations exclude some possible behaviors
of that part, then he can \emph{deduce} that those behaviors did not
occur.  His observations have \emph{disclosed} something.

These observations occur on a set of channels $C_o\subseteq\Channels$,
and the deductions constrain the events on a set of channels
$C_s\subseteq\Channels$.  $C_o$ is the set of \emph{observed}
channels, and $C_s$ is the set of \emph{source} channels.  The
observer has access to the events on the channels in $C_o$ in an
execution, using these events to learn about what happened at the
source.  The observed events may rule out some behaviors on the channels
$C_s$.

\begin{definition}
  \label{def:local}
  Let $C\subseteq\Channels$, and $\bnd\in\ES(\Frame)$.
  \begin{enumerate}
    \item The \emph{restriction} $\restrict\bnd C$ \emph{of} $\bnd$ to
    $C$ is $(B_0,R)$, where
    \begin{description}
      \item[$B_0=$] $\{e\in\bnd\colon\chan(e)\in C\}$, and 
      \item[$R=$] $(\mathop\preceq \mathrel\cap B_0\times B_0)$. 
    \end{description}
%
%


    \item $\bnd\in\ES(\Frame)$ is a $C$-\emph{run} iff for some
    $\bndA\in\exec(\Frame)$, $\bnd=\restrict{\bndA} C$.  We write
    $\lruns{C}(\Frame)$, or sometimes $\lruns{C}$, for the set of
    $C$-runs of $\Frame$.  A \emph{local run} is a member of
    $\lruns{C}$ for the relevant $C$.  

    \item $\cmpt{C}{C'}{\bnd}$ gives the ${C'}$-runs \emph{compatible
      with} a $C$-run $\bnd$:
    \ifieee\begin{eqnarray*}\else
      $$\fi 
             \cmpt{C}{C'}{\bnd} 
             \ifieee & \fi = \ifieee & \fi
                                       \{\restrict{\bndA}{C'} \colon \bndA\in
                                       \exec(\Frame)
                                       \mbox{ and }\ifieee
             \\ && \qquad\fi \restrict{\bndA} C=\bnd\}.
                   \mbox{\qquad\defsymbol}\ifieee\end{eqnarray*}\else $$ \fi
  \end{enumerate}
\end{definition}
$\restrict\bnd C\in\ES(\Frame)$ by Lemma~\ref{lemma:events}.  In
$\cmpt{C}{C'}{\bnd}$, the lower right index $C$ indicates what type of local
run $\bnd$ is.  The lower left index $C'$ indicates the type of local runs
in the resulting set.  $J$ stands for ``joint.''  $\cmpt{C}{C'}{\bnd}$
makes sense even if $C$ and $C'$ overlap, though behavior on $C\cap
C'$ is not hidden from observations at $C$.
\begin{lemma}
  \begin{enumerate}
    \item $\lruns C=\cmpt{\emptyset}{C}{\emptyset,\emptyset}$,
    i.e.~the local runs at $C$ are all those compatible with the empty
    event set $({\emptyset,\emptyset})$ at the empty set of channels.
    \item $\bnd\not\in\lruns C$ implies $\cmpt{C}{C'}{\bnd}=\emptyset$.
    \item $\bnd\in\lruns C$ implies $\cmpt{C}{C}{\bnd}=\{\bnd\}$.
    \item $\cmpt{C}{C'}{\bnd}\subseteq\lruns{C'}$.  \defend
  \end{enumerate}
\end{lemma}
\emph{No disclosure} means that any observation $\bnd$ at $C$ is
compatible with everything that could have occurred at $C'$, where
\emph{compatible} means that there is some execution that combines the
local $C$-run with the desired $C'$-run.

We summarize ``no disclosure'' by the Leibnizian slogan:
\emph{Everything possible is compossible}, ``compossible'' being his
coinage meaning possible together.  If $\bnd,\bnd'$ are each
separately possible---being $C,C'$-runs respectively---then there's an
execution $\bndA$ combining them, and restricting to each of them.
\begin{definition}~\label{def:non:disclosure} $\Frame$ has \emph{no
    disclosure from} $C$ \emph{to} $C'$ iff, for all $C$-runs $\bnd$,
    $\cmpt{C}{C'}{\bnd} = \lruns {C'}$.  \defend
%
\end{definition}

\subsection{Symmetry of disclosure}
\label{sec:nondisclosure:symmetry}
Like Shannon's mutual information and Sutherland's
non-deducibility~\cite{Sutherland86}, ``no disclosure'' is symmetric:

\begin{lemma}
  \label{lemma:symmetry} 
  \begin{enumerate}
    \item $\bnd'\in\cmpt{C}{C'}{\bnd}$ iff $\bnd\in\cmpt{C'}{C}{\bnd'}$.
    \item $\Frame$ has no disclosure from $C$ to $C'$ iff $\Frame$ has
    no disclosure from $C'$ to $C$.
  \end{enumerate}
\end{lemma}
\begin{proof}
  \textbf{1.}  By the definition, $\bnd'\in\cmpt{C}{C'}{\bnd}$ iff
  there exists an execution $\bnd_1$ such that
  $\restrict{\bnd_1}{C}=\bnd$ and $\restrict{\bnd_1}{C'}=\bnd'$.
  Which is equivalent to $\bnd\in\cmpt{C'}{C}{\bnd'}$.

  \smallskip\noindent\textbf{2.}  There is no disclosure from $C'$ to
  $C$ if for every $C$-run $\bnd$, for every $C'$-run $\bnd'$,
  $\bnd'\in\cmpt{C}{C'}{\bnd}$.  By Clause 1, this is the same as
  $\bnd\in\cmpt{C'}{C}{\bnd'}$.  \myqed
\end{proof}%
%
%
Because of this symmetry, we speak of no disclosure \emph{between} $C$
and $C'$.

$\bndA$ \emph{witnesses for} $\bnd'\in\cmpt{C}{C'}{\bnd}$ iff
$\bndA\in\exec(\Frame)$, $\bnd=\restrict\bndA{C}$, and
$\bnd'=\restrict\bndA{C'}$.  We also sometimes say that $\bnd_1$
witnesses for $\bnd'\in\cmpt{C}{C'}{\bnd}$ if $\bnd_1\in\lruns{C_1}$,
where $C\cup C'\subseteq {C_1}$, when $\restrict{\bnd_1}C=\bnd$ and
$\restrict{\bnd_1}{C'}=\bnd'$.

\begin{lemma}
  \label{lemma:cmpt:propagation} 
  \begin{enumerate}
    \item Suppose $C_0\subseteq C_1$ and $C_0'\subseteq C_1'$.  If
    $\Frame$ has no disclosure from $C_1$ to $C_1'$, then $\Frame$ has
    no disclosure from $C_0$ to $C_0'$.
    \item When $C_1,
    C_2,C_3\subseteq\Channels$, \label{clause:cmpt:propagation}
    $$ \cmpt{C_1}{C_3}{\bnd_1} \subseteq
    \bigcup_{\bnd_2\in\cmpt{C_1}{C_2}{\bnd_1}}
    \cmpt{C_2}{C_3}{\bnd_2}. \qquad\defsymbol$$
  \end{enumerate}
\end{lemma}
This is not an equality.  $\bnd_1$ and $\bnd_2$ may make incompatible
demands on a location $\ell$.  It may have endpoints on channels in
both $C_1$ and $C_3$; or paths may connect $\ell$ to both $C_1$ and
$C_3$ without traversing $C_2$.  Lemma~\ref{lemma:cut} shows that
otherwise equality holds.  See Appendix~\ref{sec:lemmas} for this
proof, and the longer subsequent proofs.
%

\subsection{The Cut Principle for Non-disclosure}
\label{sec:nondisclosure:cut}

Our key observation is that non-disclosure respects the graph
structure of a frame $\Frame$.  If $\lcut\subseteq\Channels$ is a cut
set in the undirected graph $\un\gr(\Frame)$, then disclosure from a
source set $\lsrc\subseteq\Channels$ to a sink
$\lobs\subseteq\Channels$ is controlled by disclosure to $\lcut$.  If
there is no disclosure from $\lsrc$ to $\lcut$, there can be no
disclosure from $\lsrc$ to $\lobs$.  As we will see in
Section~\ref{sec:cut:blur}, this property extends to limited
disclosure in the sense of disclosure to within a blur operator.

We view a cut as separating one set of channels as source from another
set of channels as sink.  Although it is more usual to take a cut to
separate sets of nodes than sets of channels, it is easy to transfer
between the channels and the relevant nodes.  If
$C\subseteq\Channels$, we let
$\pends(C)=\{\ell\colon\exists c\in C\qdot \sender(c)=\ell$ or
$\recipt(c)=\ell\}$; conversely, $\chans(L)=\{c\colon\sender(c)\in L$
or $\recipt(c)\in L\}$.  For a singleton set $\{\ell\}$ we suppress
the curly braces and write $\chans(\ell)$.

\begin{definition}
  \label{def:cut} 
  Let $\lsrc,\lcut,\lobs\subseteq\Channels$ be sets of channels;
  $\lcut$ is an \emph{undirected cut} (or simply a \emph{cut}) between
  $\lsrc,\lobs$ iff
  \begin{enumerate}
    \item $\lsrc,\lcut,\lobs$ are pairwise disjoint; and
    \item every undirected path $p_1$ in $\un\gr(\Frame)$ from any
    $\ell_1\in\pends(\lobs)$ to any $\ell_2\in\pends(\lsrc)$ traverses
    some member of $\lcut$.  \defend
  \end{enumerate}
\end{definition}
For instance, in Fig.~\ref{fig:net:1}, $\{c_1,c_2\}$ is a cut between
$\chans(i)$ and $\chans(\{n_1,n_2\})$.
%
%
Lemma~\ref{lemma:cut} serves as the heart of the proofs of the two
main theorems about cuts, Thms.~\ref{thm:cut:I} and~\ref{thm:cut:II}.
\begin{lemma}
  \label{lemma:cut} 
  Let $\lcut$ be an undirected cut between $\lsrc,\lobs$, and let
  $\bnd_o\in\lruns\lobs$.  Then 
  $$ \cmpt{\lobs}{\lsrc}{\bnd_o} =
  \bigcup_{\bnd_c\in\cmpt{\lobs}{\lcut}{\bnd_o}}
  \cmpt{\lcut}{\lsrc}{\bnd_c} .  \qquad\defsymbol 
  $$
\end{lemma}
\begin{proof} \emph{(Key idea.)}  First, partition $\Locations$ into
  three classes.  Let $\LEFT$ contain $\ell$ if $\ell$ has an endpoint
  on $\lobs$, or if $\ell$ can be reached by a path not traversing
  $\lcut$.  Let $\RIGHT$ contain $\ell$ if $\ell$ has an endpoint on
  $\lsrc$, or if $\ell$ can be reached by a path not traversing
  $\lcut$.  Let $\MID$ be the remainder, i.e.~locations separated from
  both $\LEFT$ and $\RIGHT$ by a channel in $\lcut$.  

  Suppose that $\bndA_1$ witnesses for
  ${\bnd_c\in\cmpt{\lobs}{\lcut}{\bnd_o}}$, and $\bndA_2$ witnesses
  for $\bnd_s\in\cmpt{\lcut}{\lsrc}{\bnd_c}$.  $\bndA_1$ and $\bndA_2$
  agree for events involving $\MID$, namely the events in $\bnd_c$
  shared between them.

  We build a witness $\bndA$ for
  ${\bnd_s\in\cmpt{\lobs}{\lcut}{\bnd_o}}$ by taking the events in
  $\bndA_1$ involving $\LEFT\cup\MID$, union the the events in
  $\bndA_2$ involving $\RIGHT\cup\MID$.  $\bndA$ is an execution
  because no location has a conflict between events from $\bndA_1$ and
  $\bndA_2$.
%
%
%
%
%
\end{proof}
The partial order semantics means that no arbitrary interleaving is
needed to create the instance $\bndA$.
%
%
Lemma~\ref{lemma:cut} is in fact a corollary of
Lemma~\ref{lemma:cut:two:frames}, which makes an analogous assertion
about a pair of overlapping frames.  
\begin{theorem}
  \label{thm:cut:I}
  Let $\lcut$ be an undirected cut between $\lsrc,\lobs$ in $\Frame$.
  If there is no disclosure between $\lsrc$ and $\lcut$, then there is
  no disclosure between $\lsrc$ and $\lobs$.
\end{theorem}
\begin{proof}
  Suppose that $\bnd_s\in\lruns\lsrc$ and $\bnd_o\in\lruns\lobs$.
  We must show  $\bnd_s\in\cmpt\lobs\lsrc{\bnd_o}$.  To apply
  Lemma~\ref{lemma:cut}, let $\bndA\in\exec(\Frame)$ such that
  $\bnd_o=\restrict{\bndA}\lobs$; $\bndA$ exists by the definition
  of $\lobs$-run.  Letting $\bnd_c=\restrict{\bndA}\lcut$, we have
  $\bnd_c\in\cmpt\lobs\lcut{\bnd_o}$.

  Since there is no disclosure between $\lcut$ and $\lsrc$,
  $\bnd_s\in\cmpt\lcut\lsrc{\bnd_c}$, and Lemma~\ref{lemma:cut}
  applies.  \myqed
\end{proof}
\begin{eg}
  In Fig.~\ref{fig:net:1} let $r_1$ be configured to discard all
  inbound packets, and $r_2$ to discard all outbound packets.  Then
  the empty event system is the only member of $\lruns{\{c_1,c_2\}}$.
  Hence there is no disclosure between $\chans(i)$ and
  ${\{c_1,c_2\}}$.  By Thm.~\ref{thm:cut:I}, there is no disclosure to
  $\chans(\{n_1,n_2\})$.  \defend
\end{eg}
So disconnected portions of a frame cannot interfere:
\begin{corollary}
  If there is no path between $\lsrc$ and $\lobs$ in
  $\un\gr(\Frame)$, then there is no disclosure between them.  
\end{corollary}
\begin{proof}
  Then $\lcut=\emptyset$ is an undirected cut set, and there is only
  one $\lcut$-run, namely the empty system of events.  It is thus
  compatible with all $\lsrc$-runs.  \myqed
\end{proof}
Thm.~\ref{thm:cut:I} and its analogue Thm.~\ref{thm:cut:II}, while
reminiscent of the max flow/min cut principle
(cf.~e.g.~\cite[Sec.~26.2]{cormen2003introduction}), are however quite
distinct from it, as the latter depends essentially on the
quantitative structure of network flows.  Our results may also seem
reminiscent of the Data Processing Inequality, stating that when three
random variables $X,Y,Z$ form a Markov chain, the mutual information
$I(X;Z)\le I(X;Y)$.  Indeed, Thm.~\ref{thm:cut:I} entails the special
case where $I(X;Y)=0$, choosing $\gr(\Frame)$ to be a single path
$X\rightarrow Y\rightarrow Z$.  For more on quantitative information
flow and our work, see the conclusion (Sec.~\ref{sec:summary}).


\section{Blur operators}
\label{sec:blur} 

We will now adapt our theory to apply to partial disclosure as well as
no disclosure.  An observer learns something about a source of
information when his observations are compatible with a proper subset
of the behaviors possible for the source.  Thus, the natural way to
measure what has been learnt is the decrease in the set of possible
behaviors at the source (see among many sources of this idea
e.g.~\cite{HalpernEtAl95,askarov2007gradual}).

This starting point suggests focusing, for every frame and regions of
interest $\lsrc\subseteq\Channels$ and $\lobs\subseteq\Channels$, on
the compatibility equivalence relations on $\lsrc$-runs:
\begin{definition}
  \label{def:obs:equiv}  
  Let $\lsrc,\lobs\subseteq\Channels$.  If $\bnd_1,\bnd_2$ are
  $\lsrc$-runs, we say that they are $\lobs$-equivalent, and write
  $\bnd_1\lobsequiv\bnd_2$, iff, for all $\lobs$-runs $\bnd_o$, 
  $$ \bnd_1\in\cmpt{\lobs}{\lsrc}{\bnd_o} \mbox{ iff }
  \bnd_2\in\cmpt{\lobs}{\lsrc}{\bnd_o} \qquad\qquad\defsymbol
  $$
\end{definition}
\begin{lemma}  
  \label{lemma:obs:equiv}
  For each $\lobs$, $\lobsequiv$ is an equivalence relation.  \defend 
\end{lemma}
\emph{No disclosure} means that all $\lsrc$-runs are
$\lobs$-equivalent.  Any notion of partial disclosure must respect
$\lobs$-equivalence, since no observations can possibly ``split''
apart $\lobs$-equivalent $\lsrc$-runs.  Partial disclosures always
respect unions of $\lobs$-equivalence classes.  

Rather than working directly with these unions of $\lobs$-equivalence
classes, we instead focus on functions on sets of runs that satisfy
three properties.  These properties express the structural principles
on partial disclosure that make our cut-blur and compositional
principles hold.  We call operators satisfying the properties
\emph{blur operators}.  Lemma~\ref{lemma:partition} shows that an
operator that always returns unions of $\lobs$-equivalence classes is
necessarily a blur operator.

When we want to prove results about all notions of partial disclosure,
we prove them for all blur operators.  When we want to show a
particular relation is a possible notion of partial disclosure, we
show that it generates an equivalence relation;
Lemma~\ref{lemma:partition} then justifies us in applying
Thms.~\ref{thm:cut:II}, \ref{cor:compose}.
\begin{definition}
  \label{def:blur}
  A function $f$ on sets is a \emph{blur operator} iff it satisfies:
  \begin{description}
    \item[Inclusion:] For all sets $S$, $S\subseteq f(S)$;
    \item[Idempotence:] $f$ is idempotent, i.e.~for all sets $S$,
    $f(f(S))=f(S)$; and
    \item[Union:] $f$ commutes with unions:  If $S_{a\in I}$ is a
    family indexed by the set $I$, then
    \begin{equation}
      \label{eq:union}
      f(\bigcup_{a\in I} S_a)=\bigcup_{a\in I}
      f(S_a). 
    \end{equation}
  \end{description}
  $S$ is $f$\emph{-blurred} iff $f$ is a blur operator and $S=f(S)$.
  \defend
\end{definition}
By \emph{Idempotence}, $S$ is $f$-blurred iff it is in the range of
the blur operator $f$.  Since $S=\bigcup_{a\in S}\{a\}$, the
\emph{Union} property says that $f$ is determined by its action on the
singleton subsets of $S$.  Thus, \emph{Inclusion} could have said
$a\in f(\{a\})$.  

Monotonicity also follows from the \emph{Union} property; if
$S_1\subseteq S_2$, then $S_2=S_0\cup S_1$, where
$S_0=S_2\setminus S_1$.  Thus $f(S_2)=f(S_0)\cup f(S_1)$, so
$f(S_1)\subseteq f(S_2)$.
\begin{lemma}
  \label{lemma:partition} 
  Suppose that $A$ is a set, and $\mathcal{R}$ is a partition of the
  elements of $A$.  There is a unique function $f_{\mathcal{R}}$ on
  sets $S\subseteq A$ such that
  \begin{enumerate}
    \item $f_{\mathcal{R}}(\{a\})=S$ iff $a\in S$ and $S\in \mathcal{R}$;
    \item $f_{\mathcal{R}}$ commutes with unions (Eqn.~\ref{eq:union}).
  \end{enumerate}
  Moreover, $f_{\mathcal{R}}$ is a blur operator.  
\end{lemma}
\begin{proof}
  Since $S=\bigcup_{a\in S}\{a\}$, $f_{\mathcal{R}}(S)$ is uniquely
  defined by the union principle (Eqn.~\ref{eq:union}).

  \emph{Inclusion} and \emph{Union} are immediate from the form of the
  definition.  \emph{Idempotence} holds because being in the same
  $\mathcal{R}$-equivalence class is transitive.  \myqed
\end{proof}
Although every equivalence relation determines a blur operator, the
converse is not true:  Not every blur operator is of this form.  For
instance, consider the set $A=\{a,b\}$, and let $f(\{a\})=\{a\}$,
$f(\{b\})=f(\{a,b\})=\{a,b\}$.  Although this $f$ is not a useful
notion of partial disclosure---as it is not generated from any $\lobs$
equivalence---Def.~\ref{def:blur} identifies the proof principles that
make Thm.~\ref{thm:cut:II} true.

For intuition about blurs, think of blurring an image: The viewer no
longer knows the details of the scene.  The viewer knows only that it
was some scene which, when blurred, would look like this, as the
following example indicates.
\begin{eg}
  \label{eg:weather-blur}
  Imaginary Weather Forecasting Inc.~(IWF) sells tailored,
  high-resolution weather data and forecasts to airlines, airports,
  and commodity investment firms, and it sells low-resolution weather
  data more cheaply to TV and radio stations.  IWF's low-tier
  subscribers should not learn higher resolution data than they have
  paid for.  There is some disclosure about high resolution data
  because (e.g.)~when low-tier subscribers see warm temperatures, they
  know that the high-resolution data is inconsistent with snow.  We
  can formalize this partial disclosure as a blur.

  Suppose IWF creates its low-resolution data $d$ by applying a lossy
  compression function $\mathsf{comp}$ to high-resolution data
  $D$. Then when low-tier subscribers receive $d$, they know that the
  high-resolution data IWF measured from the environment is some
  element of $\mathsf{comp}^{-1}(d) = \{D: \mathsf{comp}(D) = d\}$.
  These sets are $f$-blurred where
  $f(\{D\}) = \{D': \mathsf{comp}(D') = \mathsf{comp}(D)\}$.  \defend
\end{eg}

%

We will study information disclosure to within blur operators $f$,
which we interpret as meaning that $\cmpt{C}{C'}{\bnd_c}$ is
$f$-blurred.  Essentially, this is an ``upper bound'' on how much
information may be disclosed when $\bnd_c$ is observed.  The observer
will know an $f$-blurred set that the behavior at $C'$ belongs to.
However, the observer cannot infer anything finer than the $f$-blurred
sets.  
\begin{definition}
  Let $\lobs,\lsrc\subseteq\Channels$, and let $f$ be a function on
  sets of $\lsrc$-runs.  

  $\Frame$ \emph{restricts disclosure from} $\lsrc$ to $\lobs$
  \emph{to within} $f$ iff $f$ is a blur operator and, for every
  $\lobs$-run $\bnd_o$, $\cmpt{\lobs}{\lsrc}{\bnd_o}$ is an
  $f$-blurred set.  

  We also say that $\Frame$ $f$\emph{-limits}
  $\lsrc$\emph{-to-}$\lobs$ \emph{flow}.  \defend
\end{definition}
At one extreme, the no-disclosure condition is also disclosure to
within a blur operator, namely the blur operator that ignores $S$ and
adds all $C'$-runs:
\begin{equation*}
  \label{eq:blur:all}
  f_{\mathsf{all}}(S) = \{ \restrict\bndA{C'}\colon \bndA\in
  \exec(\Frame) \} . 
\end{equation*}
At the other extreme, the maximally permissive security policy is
represented as disclosure to within the identity function
$f_{\mathsf{id}}(S)=S$.  The blur operator $f_{\mathsf{id}}$ also
shows that every frame restricts disclosure to within \emph{some} blur
operator.  The equivalence classes generated by compatibility are
always $f_{\mathsf{id}}$-blurred.    

$\Frame$ may $f$-limit $\lsrc$-to-$\lobs$ flow even when the
intersection ${\lobs}\cap{\lsrc}$ is non-empty, as long as $f$ is not
too fine-grained; see below e.g.~in Def.~\ref{def:purge:blur}.
\begin{eg}\label{eg:voting}
  Suppose that $\Frame$ is an electronic voting system such as
  ThreeBallot or its siblings~\cite{rivest2007three}.  Some locations
  $L_{EC}$ are run by the election commission.  We will regard the
  voters themselves as a set of locations $L_{V}$.  Each voter
  delivers a message containing, in some form, his vote for some
  candidate.

  The election officials observe the channels connected to $L_{EC}$,
  i.e.~$\chans(L_{EC})$.  To determine the correct outcome, they must
  infer a property of the local run at $\chans(L_{V})$, namely, how
  many votes for each candidate occurred.  However, they should not
  find out which voter voted for which
  candidate~\cite{delaune2009verifying}.

  We formalize this via a blur operator.  Suppose
  $\bnd'\in\lruns{\chans(L_{V})}$ is a possible behavior of all voters
  in $L_{V}$.  Suppose that $\pi$ is a permutation of $L_{V}$.  Let
  $\pi\cdot \bnd'$ be the behavior in which each voter $\ell\in L_{V}$
  casts not his own actual vote, but the vote actually cast by
  $\pi(\ell)$.  That is, $\pi$ represents one way of reallocating the
  actual votes among different voters.  Now for any
  $S\subseteq\lruns{\chans(L_{V})}$ let
  \begin{equation}
    \label{eq:voting:blur}
    f_0(S)=\{\pi\cdot \bnd'\colon \bnd'\in S\land \pi 
    \mbox{ is a permutation of }L_{V}\}.
  \end{equation}
  This is a blur operator: (i) the identity is a permutation; (ii)
  permutations are closed under composition; and (iii)
  Eqn.~\ref{eq:voting:blur} implies commutation with unions.  The
  election commission should learn nothing about the votes of
  individuals, meaning that, for any $\bnd\in\lruns{\chans(L_{EC})}$
  the commission could observe,
  $\cmpt{\chans(L_{EC})}{\chans(L_{V})}{\bnd}$ is $f_0$-blurred.
  Permutations of compatible voting patterns are also compatible.

  This example is easily adapted to other considerations.  For
  instance, the commissioners of elections are also voters, and they
  know how they voted themselves.  Thus, we could define a (narrower)
  blur operator $f_1$ that only uses the permutations that leave
  commissioners' votes fixed.  

  In fact, voters are divided among different precincts in many cases,
  and tallies are reported on a per-precinct basis.  Thus, we have
  sets $V_1,\ldots,V_k$ of voters registered at the precincts
  $P_1,\ldots,P_k$ respectively.  The relevant blur function says that
  we can permute the votes of any two voters $v_1,v_2\in V_i$ within
  the same precinct.  One cannot permute votes between different
  precincts, since that could change the tallies in the individual
  precincts.  \defend
\end{eg}
%
%
\begin{eg}
  \label{eg:filtering} 
  Suppose in Fig.~\ref{fig:net:1}:  The inbound interface from $i$ to
  router $r_1$ discards downward-flowing packets unless their source
  is an address in $i$ and the destination is an address in $n_1,n_2$.
  The inbound interface for downward-flowing to router $r_2$ discards
  packets unless the destination address is the IP for a web server
  \texttt{www} in $n_1$, and the destination port is 80 or 443, or
  else their source port is 80 or 443 and their destination port is
  $\ge 1024$.

  We filter outbound (upward-flowing) packets symmetrically.

  A packet is \emph{importable} iff its source address is in $i$ and
  either its destination is \texttt{www} and its destination port is
  80 or 443; or else its destination address is in $n_1,n_2$, its
  source port is 80 or 443, and its destination port is $\ge 1024$.

  It is \emph{exportable} iff, symmetrically, its destination address
  is in $i$ and either its source is \texttt{www} and its source port
  is 80 or 443; or else its source address is in $n_1,n_2$, its
  destination port is 80 or 443, and its source port is $\ge 1024$.

  We will write $\select\,\bnd\,p$ for the result of selecting those
  events $e\in\events(\bnd)$ that satisfy the predicate $p(e)$,
  restricting $\preceq$ to the selected events.  Now consider the
  operator $f_i$ on $\lruns{\chans(i)}$ generated as in
  Lemma~\ref{lemma:partition} from the equivalence relation:
  \begin{description}
    \item[$\bnd_1\approx_i \bnd_2$] iff they agree on all
    \emph{importable} events, i.e.:
    \begin{eqnarray*}
      & \select  \,\bnd_1\,(\lambda e\qdot \msg(e)\mbox{ is
      importable})\cong \\ 
      & \quad \select\,\bnd_2\,(\lambda e\qdot \msg(e)\mbox{
      is importable}).
    \end{eqnarray*}
%
%
%
  \end{description}
  The router configurations mentioned above are intended to ensure
  that there is $f_i$-limited flow from $\chans(i)$ to
  $\chans(\{n_1,n_2\})$.  

  This is an \emph{integrity} condition; it is meant to ensure that
  systems in $n_1,n_2$ cannot be affected by bad (i.e.~non-importable)
  packets from $i$.

  Outbound, the blur $f_e$ on $\lruns{\chans(\{n_1,n_2\})}$ is
  generated from the equivalence relation:
  \begin{description}
    \item[$\bnd_1\approx_e \bnd_2$] iff they agree on all
    \emph{exportable} events, i.e.:
    \begin{eqnarray*}
      & \select\,\bnd_1\,(\lambda e\qdot \msg(e)\mbox{ is
        exportable}) \cong \\
      & \quad \select\,\bnd_2\,(\lambda e\qdot \msg(e)\mbox{
        is exportable}).
    \end{eqnarray*}
%
%
  \end{description}
  The router configurations are also intended to ensure that there is
  $f_e$-limited flow from $\chans(\{n_1,n_2\})$ to $\chans(i)$.  

  This is a \emph{confidentiality} condition; it is meant to ensure
  that external observers learn nothing about the non-exportable
  traffic, which was not intended to exit the organization.

  In this example, it is helpful that the exportable and importable
  packets are disjoint, and the transmission of one of these packets
  is never dependent on the reception of a packet of the other kind.
  In applications lacking this property, proving flow limitations is
  harder.  \defend
\end{eg}

%
%
%


 
%

\section{The Cut-Blur Principle}
\label{sec:cut:blur}

The symmetry of non-disclosure (Lemma~\ref{lemma:symmetry}) no longer
holds for disclosure to within a blur.  We have, however, the natural
extension of Thm~\ref{thm:cut:I}:
\begin{theorem}[Cut-Blur Principle]
  \label{thm:cut:II}
  Let $\lcut$ be an undirected cut between $\lsrc,\lobs$ in $\Frame$.
  If $\Frame$ $f$-limits $\lsrc$-to-$\lcut$ flow, then $\Frame$
  $f$-limits $\lsrc$-to-$\lobs$ flow.
\end{theorem}
\begin{proof}
  By the hypothesis, $f$ is a blur operator.  Let $\bnd_o$ be a
  $\lobs$-run.  We want to show that $\cmpt{\lobs}{\lsrc}{\bnd_o}$ is
  an $f$-blurred set,
  i.e.~$\cmpt{\lobs}{\lsrc}{\bnd_o}=f(\cmpt{\lobs}{\lsrc}{\bnd_o})$.

  For convenience, let $S_c=\cmpt{\lobs}{\lcut}{\bnd_o}$.

  By Lemma~\ref{lemma:cut}, $\cmpt{\lobs}{\lsrc}{\bnd_o} =
  \bigcup_{\bnd_c\in S_c} \cmpt{\lcut}{\lsrc}{\bnd_c}$.  Thus, we must
  show that the latter is $f$-blurred.  

  By the assumption that each $\cmpt{\lcut}{\lsrc}{\bnd_c}$ is
  $f$-blurred and by idempotence,
  $\cmpt{\lcut}{\lsrc}{\bnd_c}=f(\cmpt{\lcut}{\lsrc}{\bnd_c})$.  Now:
  \begin{eqnarray*}
    \bigcup_{\bnd_c\in S_c} \cmpt{\lcut}{\lsrc}{\bnd_c} 
    & = & \bigcup_{\bnd_c\in S_c} f(\cmpt{\lcut}{\lsrc}{\bnd_c}) \\
    & = & f(\bigcup_{\bnd_c\in S_c} \cmpt{\lcut}{\lsrc}{\bnd_c}) ,
  \end{eqnarray*}
  applying the union property (Eqn.~\ref{eq:union}).  Hence,
  $\bigcup_{\bnd_c\in S_c} \cmpt{\lcut}{\lsrc}{\bnd_c}$ is
  $f$-blurred.  \myqed
\end{proof}
This proof is the reason we introduced the \emph{Union} principle,
rather than simply considering all closure
operators~\cite{mclean1994general}.  Eqn.~\ref{eq:union} distinguishes
those closure operators that allow the ``long distance reasoning''
summarized in the proof.

\begin{eg}
  \label{eg:cut:blur}
  The frame of Example~\ref{eg:filtering} has $f_i$-limited flow from
  $\chans(i)$ to the cut $\{c_1,c_2\}$.  Thus, it has $f_i$-limited
  flow from $\chans(i)$ to $\chans(\{n_1,n_2\})$.  

  It also has $f_e$-limited flow from $\chans(\{n_1,n_2\})$ to the cut
  $\{c_1,c_2\}$.  This implies $f_e$-limited flow to $\chans(i)$.
  \defend 
\end{eg}

\subsection{A Compositional Relation between Frames}
\label{sec:cut:blur:compose}

Our next technical result gives us a way to ``transport'' a blur
security property from one frame $\Frame_1$ to another frame
$\Frame_2$.  It assumes that the two frames share a common core, some
set of locations $L_0$.  These locations should hold the same channel
endpoints in each of $\Frame_1,\Frame_2$, and should engage in the
same traces.  The boundary separating $L_0$ from the remainder of
$\Frame_1,\Frame_2$ necessarily forms a cut set $\lcut$.  Assuming
that the local runs at $\lcut$ are respected, blur properties are
preserved from $\Frame_1$ to $\Frame_2$.

\begin{definition}
  \label{def:shared:locations} 
  A set $L_0$ of locations is \emph{shared between} $\Frame_1$ and
  $\Frame_2$ iff $\Frame_1,\Frame_2$ are frames with locations
  $\Locations_1,\Locations_2$, endpoints $\pends_1,\pends_2$ and
  traces $\traces_1,\traces_2$, resp., where
  $L_0\subseteq\Locations_1\cap\Locations_2$, and for all
  $\ell\in L_0$, $\pends_1(\ell)=\pends_2(\ell)$ and
  $\traces_1(\ell)=\traces_2(\ell)$.

  When $L_0$ is shared between $\Frame_1$ and $\Frame_2$, let:
  \begin{description}
    \item[$\LEFTZ=$] $\{c\in\Channels_1\colon$ both endpoints of $c$
    are locations $\ell\in L_0\}$;
    \item[$\lcutz=$] $\{c\in\Channels_1\colon$ exactly one endpoint of
    $c$ is a location $\ell\in L_0\}$; and 
    \item[$\RIGHT_i=$] $\{c\in\Channels_i\colon$ neither endpoint of
    $c$ is a location $\ell\in L_0\}$, for $i=1,2$.  
  \end{description}
  We will also use $\lrunsone{C}$ and $\lrunstwo{C}$ to refer to the
  local runs of $C$ within $\Frame_1$ and $\Frame_2$, resp.; and
  $\cmptone{C}{C'}{\bnd}$ and $\cmpttwo{C}{C'}{\bnd}$ will refer to
  the compatible $C'$ runs in the frames $\Frame_1$ and $\Frame_2$,
  resp.  \defend
\end{definition}
Indeed, $\lcutz$ is an undirected cut between $\LEFTZ$ and $\RIGHT_i$
in $\Frame_i$, for $i=1$ and 2.  In an undirected path that starts in
$\LEFTZ$ and never traverses $\lcutz$, each arc always has both ends
in $L_0$.  We next prove a two-frame analog of Lemma~\ref{lemma:cut}.

\begin{lemma}\label{lemma:cut:two:frames}
  Let $L_0$ be shared between frames $\Frame_1,\Frame_2$.  Let
  $\lsrc\subseteq\LEFTZ$, and
  $\bnd_c\in\lrunsone\lcutz\cap\lrunstwo\lcutz$.
  \begin{enumerate}
    \item $\cmptone{\lcutz}{\lsrc}{\bnd_c} =
    \cmpttwo{\lcutz}{\lsrc}{\bnd_c}$.\label{clause:cmptone:cmpttwo}
    \item Assume
    $\lruns\lcutz(\Frame_2) \subseteq \lruns\lcutz(\Frame_1)$.  Let
    $\lobs\subseteq\RIGHT_2$, and $\bnd_o\in\lrunstwo\lobs$.
    Then\label{clause:cmptone:cmpttwo:equality}
    $$ \cmpttwo{\lobs}{\lsrc}{\bnd_o} =
    \bigcup_{\bnd_c\in\cmpttwo{\lobs}{\lcutz}{\bnd_o}}
    \cmptone{\lcutz}{\lsrc}{\bnd_c} .
    $$
  \end{enumerate}
\end{lemma}
Part~\ref{clause:cmptone:cmpttwo} states that causality acts locally.
The variable portions $\RIGHT_1,\RIGHT_2$ of $\Frame_1$ and $\Frame_2$
can affect what happens in their shared part $\LEFT$.  But it does so
only by changing which $\lruns\lcutz$ are possible.  Whenever both
frames agree on any $\bnd_c\in\lrunsone\lcutz\cap\lrunstwo\lcutz$,
then the $\lruns\LEFT$ runs compatible with $\bnd_c$ are the same.
The variation between $\RIGHT_1,\RIGHT_2$ affects the distant behavior
within $\LEFT$ only by changing the possible local runs at the
boundary $\lcutz$.  

The assumption
$\lruns\lcutz(\Frame_2) \subseteq \lruns\lcutz(\Frame_1)$ in
Part~\ref{clause:cmptone:cmpttwo:equality} and Thm.~\ref{cor:compose}
is meant to limit this variability in one direction.
\begin{theorem}
  \label{cor:compose}
  Suppose that $L_0$ is shared between frames $\Frame_1,\Frame_2$, and
  assume $\lruns\lcutz(\Frame_2) \subseteq \lruns\lcutz(\Frame_1)$.
  Consider any $\lsrc\subseteq \LEFTZ$ and $\lobs\subseteq \RIGHT_2$.
  If $\Frame_1$ $f$-limits $\lsrc$-to-$\lcutz$ flow, then $\Frame_2$
  $f$-limits $\lsrc$-to-$\lobs$ flow.
\end{theorem}
The proof is similar to the proof of the cut-blur principle, which
effectively results from it by replacing
Lemma~\ref{lemma:cut:two:frames} by Lemma~\ref{lemma:cut}, and
omitting the subscripts on frames and their local runs.  The cut-blur
principle is in fact the corollary of Thm.~\ref{cor:compose} for
$\Frame_1=\Frame_2$.

\subsection{Two Applications}
\label{sec:cut:blur:applications}

Thm.~\ref{cor:compose} is a useful compositional principle.  It
implies, for instance in connection with Example~\ref{eg:cut:blur},
that non-exportable traffic in $n_1,n_2$ remains unobservable even as
we vary the top part of Fig.~\ref{fig:net:1}.
\begin{eg}
  \label{eg:compositional} 
  Regarding Fig.~\ref{fig:net:1} as the frame $\Frame_1$, let $L_0$
  be the locations below $\{c_1,c_2\}$, and let $\lcut=\{c_1,c_2\}$.
  Let $\Frame_2$ contain $L_0,\lcut$ as shown, and have any graph
  structure above $\lcut$ such that $\lcut$ remains a cut between the
  new structure and $\Frame_0$.  Let the new locations have any
  transition systems such that the local runs agree,
  i.e.~$\lruns\lcut(\Frame_2)=\lruns\lcut(\Frame_1)$.
  Then by Thm.~\ref{cor:compose}, external observations of
  $\chans(\{n_1,n_2\})$ are guaranteed to blur out non-exportable
  events.  \defend
\end{eg}
It is appealing that our security goal is independent of changes in
the structure of the internet that we do not control.  A similar
property holds for the integrity goal of Example~\ref{eg:cut:blur} as
we alter the internal network.  The converse questions---preserving
the confidentiality property as the internal network changes, and the
integrity property as the internet changes---appear to require a
different theorem; cf.~\cite{GuttmanRowe14:cut}.

\begin{eg}
  \label{eg:voting:compositional}
  Consider a frame $\Frame_1$ representing a precinct, as shown in
  Fig.~\ref{fig:precinct}.  It consists of a set of voters
  $\overline v=\{v_1,\ldots,v_k\}$, a ballot box $\mathit{BB}_1$, and
  a channel $c_1$ connecting that to the election commission
  $\mathit{EC}$.  The $\mathit{EC}$ publishes the results over the
  channel $p$ to the public ${\mathit{Pub}}$.

  We have proved that a particular implementation of $\mathit{BB}_1$
  ensures that $\Frame_1$ blurs the votes; we formalized this within
  the theorem prover PVS.  That is, if a pattern of voting in precinct
  $i$ is compatible with an observation at $c_1$, then any permutation
  of the votes at $\overline v$ is also compatible.

  The cut-blur principle implies this blur also applies to
  observations at channel $p$ to the public.  Other implementations of
  $\mathit{BB}_1$ also achieve this property.  ThreeBallot and
  VAV~\cite{rivest2007three} appear to have this effect; they involve
  some additional data delivered to ${\mathit{Pub}}$, namely the
  receipts for the ballots.\footnote{Our claim is non-probabilistic.
    For quantitative conclusions, this no longer holds:  Some
    permutations are more likely than others, given the
    receipts~\cite{clark2007security,MoranHeatherSchneider2015}.}

  However, elections generally concern many precincts.  Frame
  $\Frame_2$ contains $i$ precincts, all connected to the election
  commission $\mathit{EC}$ (Fig.~\ref{fig:election}).  Taking
  $L_0=\overline v\cup\{\mathit{BB}_1\}$, we may apply
  Thm.~\ref{cor:compose}.  We now have $\lcut=\{c_1\}$.  Thus, to
  infer that $\Frame_2$ blurs observations of the voters in precinct
  1, we need only check that $\{c_1\}$ has no new local runs in
  $\Frame_2$.

  By symmetry, each precinct in $\Frame_2$ enjoys the same blur.

  Thus---for a given local run at $p$---any permutation of the votes
  at $\overline v$ preserves compatibility in $\Frame_2$, and any
  permutation of the votes at $\overline w$ preserves compatibility in
  $\Frame_2$.  However, Thm.~\ref{cor:compose} does not say that any
  pair of permutations at $\overline v$ and $\overline w$ must be
  jointly compatible.  That is, does every permutation on
  $\overline v\cup\overline w$ that respects the division between the
  precinct of the $\overline v$s and the precinct of the
  $\overline w$s preserve compatibility?  Although
  Thm.~\ref{cor:compose} does not answer this question, the answer is
  yes, as we can see by applying
  Lemma~\ref{lemma:combine:executions:two:frames} to $\Frame_2$.
 \defend
\end{eg}
Thm.~\ref{cor:compose} is a tool to justify abstractions.
Fig.~\ref{fig:net:1} is a sound abstraction of a variety of networks,
and Fig.~\ref{fig:precinct} is a sound abstraction of the various
multiple precinct instances of Fig.~\ref{fig:election}.




\section{Relating Blurs to Noninterference and Nondeducibility} 
\label{sec:blur:non:interference}

If we specialize frames to state machines (see
Fig.~\ref{fig:monolith}), we can reproduce some of the traditional
definitions.  Let $D=\{d_1,\ldots,d_k\}$ be a finite set of
\emph{domains}, i.e.~sensitivity labels;
$\mathop{\hookrightarrow}\mathrel\subseteq D\times D$ specifies which
domains \emph{may influence} each other.  We assume $\hookrightarrow$
is reflexive, though not necessarily transitive.
$A$ is a set of \emph{actions}, and $\dom\colon A\rightarrow D$
assigns a domain to each action; $O$ is a set of outputs.

$M=\seq{S,s_0,A,\delta,\lobs}$ is a (possibly non-deterministic) state
machine with states $S$, initial state $s_0$, transition relation
$\delta\subseteq S\times A\times S$, and observation function
$\lobs\colon S\times D\rightarrow O$.  $M$ has a set of traces
$\alpha$, and each trace $\alpha$ determines a sequence of
observations for each
domain~\cite{rushby1992noninterference,van2007indeed,van2012architectural}.

$M$ accepts commands from $A$ along the incoming channels
$c_i^{\mathsf{in}}$ from the $d_i$; each command $a\in A$ received
from $d_i$ has sensitivity $\dom(a)=d_i$.  $M$ delivers observations
along the outgoing channels $c_i^{\mathsf{out}}$.  The frame requires
a little extra memory, in addition to the states of $M$, to deliver
outputs over the channels $c_i^{\mathsf{out}}$.

$\Frame$ being star-like, each $\bndA\in\exec(\Frame)$ has a linearly
ordered $\preceq_\bndA$, since all events are in the linearly ordered
$\proj(\bndA,M)$.
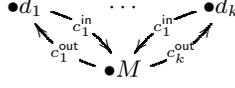
\begin{figure}[tb]
  $$\xymatrix@C=6mm@R=4mm{
    \bullet d_1\ar@/^/[dr]|{c_{1}^{\mathsf{in}}} & \cdots & \bullet
    d_k\ar@/_/[dl]|{c_{1}^{\mathsf{in}}} \\
    & \bullet
    M\ar@/^/[ul]|{c_{1}^\mathsf{out}}\ar@/_/[ur]|{c_{k}^\mathsf{out}}
    & }
  $$
  \caption{Machine $M$, domains $\{d_1,\ldots,d_k\}$}
  \label{fig:monolith}
\end{figure}
Let us write:
\begin{quote}
  \begin{description}
    \item[$C_i=\{c_{i}^{\mathsf{in}},c_{i}^{\mathsf{out}}\}$] for
    $d_i$'s input and output for $M$; 
    \item[${\vis(d_i)}=\{c_j^{\mathsf{in}}\colon d_j\hookrightarrow
    d_i\}$] for inputs visible to $d_i$;
    \item[${\mathsf{IN}}={\{{c_{x}^{\mathsf{in}}}\colon 1\le x\le
      k\}}$] for the input channels;
    \item[$\inp(\bndA)=\restrict\bndA{\mathsf{IN}}$] for all input
    behavior in $\bndA$.
  \end{description}
\end{quote}

\paragraph{Noninterference and nondeducibility.}
\emph{Noninterference}~\cite{GoguenMeseguer82} and its variants are
defined by \emph{purge} functions $p$ for each target domain $d_i$,
defined by recursion on input behaviors $\inp(\bndA)$.  The original
Goguen-Meseguer (GM) purge function $p_o$ for
$d_i$~\cite{GoguenMeseguer82} retains the events $e\in\inp(\bndA)$
satisfying the predicate
$$ \chan(e)\in{\vis(d_i)} .  $$
A purge function for intransitive $\hookrightarrow$ relations was
subsequently proposed by Haigh and Young~\cite{HaighYoung1987}.  In
the purge function for domain $d_i$, any input event
$e_0\in\inp(\bndA)$ is retained if $\inp(\bndA)$ has an increasing
subsequence $e_0\preceq e_1\preceq\ldots\preceq e_{j}$ where
$\dom(\chan(e_{j}))=d_i$ and, for each $k$ with $0\le k <j$,
$$ \chan(e_k) \in{\vis(\dom(e_{k+1}))} .
$$
%
In~\cite{van2007indeed}, van der Meyden's purge functions yield tree
structures instead of subsequences; every path from a leaf to the root
in these trees is a subsequence consisting of permissible effects
$\chan(e_k) \in{\vis(\dom(e_{k+1}))}$.  This tightens the
notion of security, because the trees do not include ordering
information between events that lie on different branches to the root.

We formalize a \emph{purge function} for a domain $d_i\in D$ as being
a function from executions $\bndA$ to some range set $A$.  It should
be sensitive only to \emph{input} events in $\bndA$
(condition~\ref{clause:inputs:only}), and it should certainly reflect
\emph{all} the inputs \emph{visible} to level $d_i$
(condition~\ref{clause:visible:inputs}).  In most existing
definitions, the range $A$ consists of sequences of input events,
though in van der Meyden's~\cite{van2007indeed}, they are trees of
input events.  In~\cite{ChongvanderMeyden14}, the range depends on how
declassification conditions are defined.
\begin{definition}\label{def:purge}
  Let $\Frame$ be as in Fig.~\ref{fig:monolith}, and $A$ any set.  A
  function $p\colon\exec(\Frame)\rightarrow A$ is a $d_i$-\emph{purge
    function}, where $d_i\in D$, iff
    \begin{enumerate}
      \item $\inp(\bndA)=\inp(\bndA')$\label{clause:inputs:only}
      implies $p(\bndA)=p(\bndA')$;
      \item $p(\bndA)=p(\bndA')$\label{clause:visible:inputs} implies
      $\restrict\bndA{\vis(d_i)}=\restrict{\bndA'}{\vis(d_i)}$.
    \end{enumerate}
    If $p$ is a $d_i$-purge, $\bndA\fequiv\bndA'$ means
    $p(\bndA)=p(\bndA')$.  \defend
\end{definition}
Each purge $p$ determines notions of noninterference and
nondeducibility.
\begin{definition}  \label{def:noninterference}
  Let $p$ be a purge function for $d_i\in D$.  $\Frame$ \emph{is
    $p$-noninterfering}, written $\Frame\in\NI^p$, iff, for all
  $\bndA,\bndA'\in\exec(\Frame)$,
    $$\bndA\fequiv\bndA' \mbox{ implies }
    \restrict\bndA{C_i}=\restrict{\bndA'}{C_i}.  $$
    \item $\Frame$ \emph{is $p$-nondeducible} ($\Frame\in\ND^p$), iff,
    for all $\bndA,\bndA'\in\exec(\Frame)$,
    $$\bndA\fequiv\bndA' \mbox{ implies }
    \restrict{\bndA'}{\mathsf{IN}} \in
    \cmpt{C_i}{{\mathsf{IN}}}{\restrict\bndA{C_i}}. 
    \qquad\defsymbol
    $$
\end{definition}
%
%
Here we take non-deducibility to mean that $d_i$'s observations
provide no more information about all inputs than the purge $p$
preserves.  Thus, ${\restrict\bndA{C_i}}$ is akin to Sutherland's
\emph{view}~\cite[Sec.~5.2]{Sutherland86}.  We have adapted it
slightly.  Sutherland assumes that $d_i$ also observes the outputs on
the channels $c_{j}^{\mathsf{out}}$ where $d_j\hookrightarrow d_i$; we
instead let the definition of machine $M$ replicate them onto
$c_{i}^{\mathsf{out}}$ when desired.

Sutherland's \emph{hidden\_from} could be interpreted as
$\restrict{\bndA'}{\{c_j^{\mathsf{in}}\colon d_j\not\hookrightarrow
  d_i\}}$,
i.e.~the inputs that would not be visible to $d_i$.  This agrees with
our proposed definition in the case Sutherland considered, namely the
classic GM purge for noninterference.  The assumption
$\bndA\fequiv\bndA'$ is meant to extend nondeducibility for other
purges.  As expected, noninterference is tighter than
nondeducibility~\cite[Sec.~7]{Sutherland86}:
\begin{lemma}\label{lemma:ni:nd}
  Let $p$ be a purge function for domain $d_i$.  $\Frame\in\NI^p$
  implies $\Frame\in\ND^p$.
\end{lemma}
\begin{proof}
  Assume that $\Frame\in\NI^p$ and $\bndA,\bndA'\in\exec(\Frame)$,
  where $\bndA\fequiv\bndA'$.  

  By the definition, $\restrict\bndA{C_i}=\restrict{\bndA'}{C_i}$.
  Thus,
  $\cmpt{C_i}{\mathsf{IN}}{\restrict\bndA{C_i}} =
  \cmpt{C_i}{\mathsf{IN}}{\restrict{\bndA'}{C_i}}$. 

  But
  $\restrict{\bndA'}{\mathsf{IN}}\in\cmpt{C_i}{\mathsf{IN}}{\restrict{\bndA'}{C_i}}$
  because $\bndA'$ is itself a witness.  \myqed
\end{proof}
$\NI^p$ and $\ND^p$ are not equivalent, as $\ND^p$ has an additional
(implicit) existential quantifier.  The witness execution showing that
$\restrict{\bndA'}{\mathsf{IN}}\in\cmpt{C_i}{\mathsf{IN}}{\restrict\bndA{C_i}}$
may differ from $\bndA'$ on channels $c\not\in{\mathsf{IN}}\cup{C_i}$,
namely the output channels $c_{j}^{\mathsf{out}}$ for $j\not=i$.

The symmetry of nondisclosure (Lemma~\ref{lemma:symmetry}) does not
hold for $\NI^p$ and $\ND^p$.  For instance, relative to the GM purge
for flow to $d_i$, there may be noninterference for inputs at $d_j$,
while there is interference for flow from $d_i$ to $d_j$.  The
asymmetry arises because the events to be concealed are only inputs at
the source, while the observed events are both inputs and
outputs~\cite{Sutherland86}.  

The idea of $p$-noninterference is useful only when $M$ is
deterministic, since otherwise the outputs observed on
$c_{i}^{\mathsf{out}}$ may differ even when
$\inp(\bndA)=\inp(\bndA')$.  For non-deterministic $M$,
nondeducibility is more natural.

\paragraph{Purges and blurs.}  We can associate a blur operator $f^p$
with each purge function $p$, such that $\ND^p$ amounts to respecting
the blur operator $f^p$.  We regard $\ND^p$ as saying that the
input/output events on $C_i$ tell $d_i$ no more about all the inputs
than the purged input $p(\bndA)$ would disclose.  We use a
compatibility relation where the observed channels and the source
channels overlap on $c_{i}^{\mathsf{in}}$.
\begin{definition} \label{def:purge:blur} 
  Let $p$ be a purge function for $d_i$, and define the equivalence
  relation
  ${\mathcal{R}}\subseteq(\lruns{\mathsf{IN}}\times\lruns{\mathsf{IN}})$
  by the condition: $ {\mathcal{R}}(\bnd_1,\bnd_2)$ iff
  \begin{equation}
    \label{eq:nd:blur}
    \exists\bndA_1,\bndA_2\in\exec(\Frame)\qdot
    \bigwedge_{j=1,2} \bnd_j=\restrict{\bndA_j}{\mathsf{IN}} \land
    \bndA_1\fequiv\bndA_2 .
  \end{equation}
  Define
  $f^p \colon \mathcal{P}(\lruns{\mathsf{IN}}) \rightarrow
  \mathcal{P}(\lruns{\mathsf{IN}}) $
  to close under the ${\mathcal{R}}$-equivalence classes as in
  Lemma~\ref{lemma:partition}.  \defend
\end{definition}
%
In fact, $\ND$ is a form of disclosure limited to within a blur:
\begin{lemma} \label{lemma:nd:blur} 
  Let $p$ be a purge function for domain $d_i$.  For all $\Frame$,
  $\Frame\in\ND^p$ iff $\Frame$ $f^p$-limits ${\mathsf{IN}}$-to-$C_i$
  flow.
\end{lemma}
\begin{proof}
%
  \textbf{1.  $\ND^p$ implies $f^p$-limited flow.} Suppose that
  $\Frame\in\ND^p$; $\bnd_i\in\lruns{C_i}$; and
  $\bnd_1\in\cmpt{C_i}{\mathsf{IN}}{\bnd_i}$.  If $\bnd_2\in
  f^p(\bnd_1)$, we must show that $\bnd_2\in
  \cmpt{C_i}{\mathsf{IN}}{\bnd_i}$. 

  By Def.~\ref{def:purge:blur} there are $\bndA_1,\bndA_2$ so that
  $\restrict{\bndA_1}{\mathsf{IN}} = \bnd_1$ and
  $\restrict{\bndA_2}{\mathsf{IN}} = \bnd_2$ and
  $\bndA_1\fequiv\bndA_2$. Furthermore, let $\bndA$ witness
  $\bnd_1\in\cmpt{C_i}{\mathsf{IN}}{\bnd_i}$. Then
  $\restrict{\bndA}{\mathsf{IN}} = \bnd_1 =
  \restrict{\bndA_1}{\mathsf{IN}}$. So Def.~\ref{def:purge},
  Clause~\ref{clause:inputs:only} says $\bndA\fequiv\bndA_1$ and so
  also $\bndA\fequiv\bndA_2$. Since $\Frame\in\ND^p$,
  $\restrict{\bndA_2}{\mathsf{IN}} \in
  \cmpt{C_i}{\mathsf{IN}}{\restrict{\bndA}{C_i}}$. That is, $\bnd_2
  \in \cmpt{C_i}{\mathsf{IN}}{\bnd_i}$ as required.

  \smallskip\noindent\textbf{2.  $f^p$-limited flow implies $\ND^p$.}
  Assume $\cmpt{C_i}{\mathsf{IN}}{\bnd_i}$ is $f^p$-blurred for all
  $\bnd_i$.  We must show that, for all $\bndA_1,\bndA_2$ such that
  $\bndA_1\fequiv\bndA_2$,
  $(\restrict{\bndA_2}{\mathsf{IN}})\in\cmpt{C_i}{\mathsf{IN}}{\restrict{\bndA_1}{C_i}}$.

  So choose executions with $\bndA_1\fequiv\bndA_2$.  By
  Def.~\ref{def:purge:blur},
  ${\mathcal{R}}(\restrict{\bndA_1}{\mathsf{IN}},\restrict{\bndA_2}{\mathsf{IN}})$,
  since $\bndA_1,\bndA_2$ satisfy the condition.  Thus, since
  $\cmpt{C_i}{\mathsf{IN}}{\restrict{\bndA_1}{C_i}}$ is $f^p$-blurred
  and contains $\restrict{\bndA_1}{\mathsf{IN}}$,
  $\restrict{\bndA_2}{\mathsf{IN}}\in\cmpt{C_i}{\mathsf{IN}}{\restrict{\bndA_1}{C_i}}$.
  \myqed
\end{proof}
%
%
%
%
\nocite{wittbold1990information}
\paragraph{Semantic sensitivity.}  Blur operators provide an explicit
semantic representation of the information that will not be disclosed
when flow is limited.  This is in contrast to intransitive
non-interference~\cite{rushby1992noninterference,HaighYoung1987,van2007indeed},
which considers only whether the ``$\hookrightarrow$ plumbing'' among
domains is correct.
\begin{eg}
  \label{eg:weather-compare}
  We represent Imaginary Weather Forecasting (IWF, see
  Example~\ref{eg:weather-blur}) as a state machine frame as in
  Fig.~\ref{fig:monolith}.  It has domains
  $\{\mathit{ws},\ell,p,\mathit{cmp}\}$ for the weather service,
  low-tier customer, premium-tier customer and compression service
  respectively. Let $\hookrightarrow$ be the smallest reflexive (but
  intransitive) relation extending Eqn.~\ref{eq:weather:flow}, where
  all reports must flow through the compression service:
  \begin{equation}
    \label{eq:weather:flow}
    \mathit{ws}\hookrightarrow\mathit{cmp}\hookrightarrow p\mbox{ and
    } \mathit{cmp}\hookrightarrow \ell .  
  \end{equation}

  The $\mathit{cmp}$ service should compress reports lossily before
  sending them to $\ell$ and compress them losslessly for $p$.
  However, a faulty $cmp$ may compress losslessly for both $\ell$ and
  $p$.  Purge functions~\cite{HaighYoung1987,van2007indeed} do not
  distinguish between correct and faulty $\mathit{cmp}$s.  In both
  cases, all information from $\mathit{ws}$ does indeed pass through
  $\mathit{cmp}$.  The blur of Example~\ref{eg:weather-blur}, however,
  defines the desired goal semantically.  With the faulty
  $\mathit{cmp}$, the high-resolution data compatible with the
  observation of $\ell$ is more sharply defined than an $f$-blurred
  set.  \defend
\end{eg}


\section{Future Work}
\label{sec:summary}

We have explored how the graph structure of a distributed system helps
to constrain information flow.  We have established the cut-blur
principle.  It allows us to propagate conclusions about limited
disclosure from a cut set $\lcut$ to more remote parts of the graph.
We have also showed a sufficient condition for limitations on
disclosure to be preserved under homomorphisms.  Most of our examples
concern networks and filtering, but the ideas are much more widely
applicable.

\paragraph{Quantitative treatment.}  It should be possible to equip
frames with a quantitative information flow semantics.  One obstacle
here is that our execution model mixes some choices which are natural
to view probabilistically---for instance, selection between different
outputs when both are permitted by an \textsc{lts}---with others that
seem non-deterministic.  The choice between receiving an input and
emitting an output is an example of this, as is the choice between
receiving inputs on different channels.  This problem has been studied
(e.g.~\cite{CanettiCKLLPS08,chatzikokolakis2007making}), but a
tractable semantics may require new ideas.

\paragraph{A Dynamic Model.}
\label{sec:summary:dynamic}
Instead of building $\pends(\ell)$ into the frame, so that it remains
fixed through execution, we may alternatively regard it as a component
of the states of the individual locations.  Let us regard
$\traces(\ell)$ as generated by a labeled transition system
$\lts\ell$.  Then we may enrich the labels $c,v$
%
%
so that they also involve a sequence of endpoints
$\overline p\subseteq \Endpoints$:
$$(c,v,\overline p).  
$$
The transition relation of $\lts\ell$ is then constrained to allow a
transmission $(c,v,\overline p)$ in a state only if
$p\subseteq\pends(\ell)$ holds in that state, in which case
$\overline p$ is omitted in the next state.  A reception
$(c,v,\overline p)$ causes $\overline p$ to be added to the next state
of the receiving location.

The cut-blur principle remains true in an important case:  A set
$\lcut$ is an \emph{invariant cut} between $\lsrc$ and $\lobs$ if it
is an undirected cut, and moreover the execution of the frame
preserves this property.  Then the cut-blur principle holds in the
dynamic model for invariant cuts.  


This dynamic model suggests an analysis of security-aware software
using object capabilities.  Object capabilities may be viewed as
transmission endpoints.  To use a capability, one transmits a message
down the channel to the object itself, which holds the reception
endpoint.  To transfer a capability, one transmits its
transmission endpoint down another channel.

McCamant and Ernst~\cite{mccamant2008quantitative}'s quantitative
approach generates a directed graph of this sort in memory at runtime.
Providing a maximum over all possible runs would appear to depend on
inferring some invariants on the structure of the graphs.  Our methods
might be helpful for this.

\paragraph{Cryptographic Masking.}  Encryption is not a blur.
Encrypting messages makes their contents unavailable in locations
lacking the decryption keys.  In particular, locations lacking the
decryption key may form a cut set between the source and destination
of the encrypted message.  However, at the destinations, where the
keys are available, the messages can be decrypted and their contents
observed.  Thus, the cut-blur theorem says it would be wrong to view
encryption as a blur in this set-up:  Its effects can be undone beyond
the cut.

Instead, we enrich the model.  The messages on a channel are not
observed in their specificity.  Instead, the observation is a
projection of messages on the channel.  One could represent this
projection using the Abadi-Rogaway~\cite{AbadiRogaway02} pattern
operator $\squareforqed$.  Alternatively, one could rely directly on
the cryptographic indistinguishability
relation~\cite{askarov2008cryptographically}.

In general, this means applying the compatibility function
$\cmpt{C}{C'}{\bnd}$ not only to local $\lruns C$ $\bnd$, but also to
their projections $\cmpt{C}{C'}{\mathsf{proj}(\bnd)}$ for projection
functions of interest.  This change is quite compatible with many of
our proof methods, but the statements of theorems will need
adaptation.  We would like to use the resulting set-up to reason about
cryptographic voting systems, such as Helios and
Pr{\^e}t-{\`a}-Voter~\cite{Adida2008,ryan2009pret}.

We also intend to provide tool support for defining relevant blurs and
establishing that they limit disclosure in several application areas.


\paragraph{Acknowledgments.}  We are grateful to Megumi Ando, Aslan
Askarov, Stephen Chong, John Ramsdell, and Mitchell Wand.  In
particular, John Ramsdell formalized Thms.~\ref{thm:cut:II}
and~\ref{cor:compose} in PVS, as well as Example~\ref{eg:voting}; this
suggested Example~\ref{eg:voting:compositional}.

\bibliography{../../../bibtex/secureprotocols}
\bibliographystyle{plain}

\appendix
\ifieee \subsection{Some Lemmas and Proofs}\else\section{Some Lemmas
  and Proofs}\fi
\label{sec:lemmas}


\begin{repeatthm}{Lemma}{lemma:cmpt:propagation}
  \begin{enumerate}
    \item Suppose $C_0\subseteq C_1$ and $C_0'\subseteq C_1'$.  If
    $\Frame$ has no disclosure from $C_1$ to $C_1'$, then $\Frame$ has
    no disclosure from $C_0$ to $C_0'$.
    \item When $C_1,
    C_2,C_3\subseteq\Channels$, 
    $$ \cmpt{C_1}{C_3}{\bnd_1} \subseteq
    \bigcup_{\bnd_2\in\cmpt{C_1}{C_2}{\bnd_1}}
    \cmpt{C_2}{C_3}{\bnd_2}.$$
  \end{enumerate}
\end{repeatthm}
\begin{proof}
  \textbf{1.}  Suppose $\bnd_0$ is a $C_0$-run, and $\bnd_0'$ is a
  $C_0'$-run.  We want to show that
  $\bnd_0'\in\cmpt{C_0}{C_0'}{\bnd_0}$.  

  Since they are local runs, there exist
  $\bndA_0,\bndA_0'\in\exec(\Frame)$ such that
  $\bnd_0=\restrict{\bndA_0}{C_0}$ and
  $\bnd_0'=\restrict{\bndA_0'}{C_0'}$.  But let
  $\bnd_1=\restrict{\bndA_0}{C_1}$ and let
  $\bnd_1'=\restrict{\bndA_0'}{C_1'}$.  By no-disclosure,
  $\bnd_1'\in\cmpt{C_1}{C_1'}{\bnd_1}$.  So there is an
  $\bndA\in\exec(\Frame)$ such that $\bnd_1=\restrict\bndA{C_1}$ and
  $\bnd_1'=\restrict\bndA{C_1'}$.

  However, then $\bndA$ witnesses for
  $\bnd_0'\in\cmpt{C_0}{C_0'}{\bnd_0}$:  After all, since
  $C_0\subseteq C_1$,
  $\restrict\bndA{C_0}=\restrict{(\restrict\bndA{C_1})}{C_0}$.
  Similarly for the primed versions.

  \smallskip\noindent\textbf{2.}  Suppose that
  $\bnd_3\in\cmpt{C_1}{C_3}{\bnd_1}$, so that there exists an
  $\bndA\in\exec(\Frame)$ such that $\bnd_1=\restrict\bndA{C_1}$ and
  $\bnd_3=\restrict\bndA{C_3}$.  Letting $\bnd_2=\restrict\bndA{C_2}$,
  the execution $\bndA$ ensures that
  $\bnd_2\in\cmpt{C_1}{C_2}{\bnd_1}$ and
  $\bnd_3\in\cmpt{C_2}{C_3}{\bnd_2}$.  \myqed
\end{proof}

We now consider different frames $\Frame_1,\Frame_2$ that overlap on a
common subset $L_0$, and show how local runs in the two can be pieced
together.  In this context, we use the notation of
Def.~\ref{def:shared:locations}.  
%
%

\begin{lemma}\label{lemma:combine:executions:two:frames}
  Let $L_0$ be shared between frames $\Frame_1,\Frame_2$.  Let
  $$\bnd_{lc}\in\lrunsone{(\LEFT\cup\lcut)} \mbox{ and }
  \bnd_{rc}\in\lrunstwo{(\RIGHT_2\cup\lcut)}$$ 
  agree on $\lcut$,
  i.e.~$\restrict{\bnd_{lc}}{\lcut} = \restrict{\bnd_{rc}}{\lcut}$.
  Then there is an $\bndA\in\exec(\Frame_2)$ such that
  $$\bnd_{lc}=\restrict\bndA{(\LEFT\cup\lcut)} \mbox{ and } \bnd_{rc}= 
  \restrict\bndA{(\RIGHT_2\cup\lcut)}.$$
\end{lemma}
\begin{proof}
  Since $\bnd_{lc}$ and $\bnd_{rc}$ are local runs of
  $\Frame_1,\Frame_2$ resp., they are restrictions of executions, so
  choose $\bndA_1\in\exec(\Frame_1)$ and $\bndA_2\in\exec(\Frame_2)$
  so that $\bnd_{lc}=\restrict{\bndA_1}{(\LEFT\cup\lcut)}$ and
  $\bnd_{rc}=\restrict{\bndA_2}{(\RIGHT_2\cup\lcut)}$.
  Now define $\bndA$ by stipulating:
  \begin{eqnarray}
    \events(\bndA) &=& \events(\bnd_{lc}) \cup
                       \events({\bnd_{rc}}
                       )  \\
    \mathop{\preceq_\bndA} &=& \mbox{\ifieee\small\fi the least
                               partial order extending
                               } \label{clause:union:preceq:two:frames}    
                               {\preceq_{\bnd_{lc}}}
                               \mathrel\cup{\preceq_{\bnd_{rc}}}
                               \ifieee\!\fi . 
  \end{eqnarray}
  Since $\bndA_1,\bndA_2$ agree on $\lcut$,
  $\events(\bndA)=\events(\restrict{\bnd_{lc}}{\LEFT}) \cup
  \events({\bnd_{rc}})$, and we could have used the latter as an
  alternate definition of $\events(\bndA)$, as well as the symmetric
  restriction of ${\bnd_{rc}}$ to $\RIGHT_2$ leaving ${\bnd_{lc}}$
  whole.

  The definition of ${\preceq_\bndA}$ as a partial order is sound,
  because there are no cycles in the union
  (\ref{clause:union:preceq:two:frames}).  Cycles would require
  $\bndA_1$ and $\bndA_2$ to disagree on the order of events in their
  restrictions to $\lcut$, contrary to assumption.  Likewise, the
  finite-predecessor property is preserved:  $x_0\preceq_\bndA x_1$
  iff $x_0,x_1$ belong to the same $\bnd_{?c}$ and are ordered there,
  or else there is an event in $\restrict{\bnd_{?c}}\lcut$ which comes
  between them.  So the events preceding $x_1$ form the finite union
  of finite sets.  Thus, $\bndA\in\ES(\Frame_2)$.

  Moreover, $\bndA$ is an execution $\bndA\in\exec(\Frame_2)$:  If
  $\ell\in L_0$, then $\proj(\bndA,\ell)=\proj({\bnd_{lc}},\ell)$, and
  the latter is a trace in $\traces_1(\ell)=\traces_2(\ell)$.  If
  $\ell\not\in L_0$, then $\proj(\bndA,\ell)=\proj({\bnd_{rc}},\ell)$,
  and the latter is a trace in $\traces_2(\ell)$.

  There is no $\ell$ with channels in both $\LEFT$ and $\RIGHT_2$.
  \myqed
\end{proof}
What makes this proof work?  Any one location either has all of its
channels lying in $\LEFTZ\cup\lcutz$ or else all of them lying in
$\RIGHT_i\cup\lcut$.  When piecing together the two executions
$\bndA_1,\bndA_2$ into a single execution $\bndA$, no location needs
to be able to execute a trace that comes partly from $\bndA_1$ and
partly from $\bndA_2$.  This is what determines our definition of cuts
using the undirected graph $\un\gr(\Frame)$.

We next prove the two-frame analog of Lemma~\ref{lemma:cut}.

\begin{repeatthm}{Lemma}{lemma:cut:two:frames}
  Let $L_0$ be shared between frames $\Frame_1,\Frame_2$.  Let
  $\lsrc\subseteq\LEFT$, and
  $\bnd_c\in\lrunsone\lcutz\cap\lrunstwo\lcutz$.
  \begin{enumerate}  
    \item $\cmptone{\lcutz}{\lsrc}{\bnd_c} =
    \cmpttwo{\lcutz}{\lsrc}{\bnd_c}$.\label{clause:cmptone:cmpttwo}
    \item Assume
    $\lruns\lcutz(\Frame_2) \subseteq \lruns\lcutz(\Frame_1)$.  Let
    $\lobs\subseteq\RIGHT_2$, and $\bnd_o\in\lrunstwo\lobs$.
    Then\label{clause:cmptone:cmpttwo:equality}
    $$ \cmpttwo{\lobs}{\lsrc}{\bnd_o} =
    \bigcup_{\bnd_c\in\cmpttwo{\lobs}{\lcutz}{\bnd_o}}
    \cmptone{\lcutz}{\lsrc}{\bnd_c} .
    $$
  \end{enumerate}
\end{repeatthm}
\begin{proof} \textbf{\ref{clause:cmptone:cmpttwo}.}  First, we show
  that $\bnd_s\in\cmptone{\lcutz}{\lsrc}{\bnd_c}$ implies
  $\bnd_s\in\cmpttwo{\lcutz}{\lsrc}{\bnd_c}$.

  Let $\bndA_1$ witness for $\bnd_s\in\cmptone{\lcut}{\lsrc}{\bnd_c}$,
  and let $\bndA_2$ witness for $\bnd_c\in\lrunstwo{\lcut}$.  Define
  $$\bnd_{lc}=\restrict{\bndA_1}{(\LEFT\cup\lcut)} \mbox{ and }
  \bnd_{rc}=\restrict{\bndA_2}{(\RIGHT_2\cup\lcut)} . $$ 
  Now the assumptions for
  Lemma~\ref{lemma:combine:executions:two:frames} are satisfied.  So
  let $\bndA\in\exec(\Frame_2)$ restrict to $\bnd_{lc}$ and
  $\bnd_{rc}$ as in the conclusion.  Thus,
  $\restrict\bndA{\lsrc}=\bnd_s$.

  For the converse, we rely on the symmetry of ``$L_0$ is shared
  between frames $\Frame_1,\Frame_2$.''

  \medskip\noindent \textbf{\ref{clause:cmptone:cmpttwo:equality}.}
  By the assumption, whenever
  ${\bnd_c\in\cmpttwo{\lobs}{\lcut}{\bnd_o}}$, then also
  ${\bnd_c\in\lrunsone{\lcut}}$.  Thus, we can apply
  part~\ref{clause:cmptone:cmpttwo} after using
  Lemma~\ref{lemma:cmpt:propagation}:
  \begin{eqnarray*}
    \cmpttwo{\lobs}{\lsrc}{\bnd_o} 
    &\subseteq& 
                \bigcup_{\bnd_c\in\cmpttwo{\lobs}{\lcutz}{\bnd_o}}
                \cmpttwo{\lcutz}{\lsrc}{\bnd_c} \\ 
    &\subseteq&
                \bigcup_{\bnd_c\in\cmpttwo{\lobs}{\lcutz}{\bnd_o}}
                \cmptone{\lcutz}{\lsrc}{\bnd_c} . 
  \end{eqnarray*}
  For the reverse inclusion, assume that
  $\bnd_s\in\cmptone{\lcutz}{\lsrc}{\bnd_c}$, where
  ${\bnd_c\in\cmpttwo{\lobs}{\lcutz}{\bnd_o}}$.  Thus, we can apply
  Lemma~\ref{lemma:combine:executions:two:frames}, obtaining
  $\bndA\in\exec(\Frame_2)$ which agrees with $\bnd_s$, $\bnd_c$, and
  $\bnd_o$.  So $\bndA$ witnesses for
  $\bnd_s\in\cmpttwo{\lobs}{\lsrc}{\bnd_o}$.  \myqed
\end{proof}
We now turn to the one-frame corollary, which we presented earlier as
Lemma~\ref{lemma:cut}.  
\begin{repeatthm}{Lemma}{lemma:cut} 
  Let $\lcut$ be an undirected cut between $\lsrc,\lobs$, and let
  $\bnd_o\in\lruns\lsrc$.  Then 
  $$ \cmpt{\lobs}{\lsrc}{\bnd_o} =
  \bigcup_{\bnd_c\in\cmpt{\lobs}{\lcut}{\bnd_o}}
  \cmpt{\lcut}{\lsrc}{\bnd_c} . 
  $$
\end{repeatthm}
\begin{proof}
  Define $L_0$ to be the smallest set of locations such that
  \begin{enumerate}
    \item $\ell\in L_0$ if $\chans(\ell)\cap\lsrc\not=\emptyset$;
    \item $L_0$ is closed under reachability by paths that do not
    traverse $\lcut$.
  \end{enumerate}
  $L_0$ is shared between $\Frame$ and itself.  Moreover, for the set
  of channels $\lcutz$ defined in Def.~\ref{def:shared:locations}, we
  have $\lcutz\subseteq\lcut$:  $\lcutz$ is the part of $\lcut$ that
  actually lies on the boundary of $L_0$.

  By Lemma~\ref{lemma:cut:two:frames}, we have 
  $$ \cmpt{\lobs}{\lsrc}{\bnd_o} =
  \bigcup_{\bnd_c\in\cmpt{\lobs}{\lcutz}{\bnd_o}}
  \cmpt{\lcut}{\lsrc}{\bnd_c} .
  $$
  Since $\lcutz\subseteq\lcut$, 
  $$  \bigcup_{\bnd_c\in\cmpt{\lobs}{\lcutz}{\bnd_o}}
  \cmpt{\lcut}{\lsrc}{\bnd_c} \subseteq
  \bigcup_{\bnd_c\in\cmpt{\lobs}{\lcut}{\bnd_o}}
  \cmpt{\lcut}{\lsrc}{\bnd_c} .
  $$
  For the converse, suppose that
  $\bnd_s\in\cmpt{\lcut}{\lsrc}{\bnd_c}$, for
  ${\bnd_c\in\cmpt{\lobs}{\lcut}{\bnd_o}}$.  Then there is $\bndA$
  such that $\restrict\bndA{\lsrc}=\bnd_s$ and
  $\restrict\bndA{\lobs}=\bnd_o$.  Thus,
  $\bnd_s\in\cmpt{\lcut}{\lsrc}{\restrict\bndA{\lcutz}}$ and
  ${\restrict\bndA{\lcutz}}\in\cmpt{\lobs}{\lcutz}{\bnd_o}$.  \myqed
\end{proof}

The cut-blur principle is also the one-frame corollary of
Thm.~\ref{cor:compose}.  The proofs are very similar.

\begin{repeatthm}{Theorem}{cor:compose}
  Suppose that $L_0$ is shared between frames $\Frame_1,\Frame_2$, and
  assume $\lruns\lcut(\Frame_2) \subseteq \lruns\lcut(\Frame_1)$.
  Consider any $\lsrc\subseteq \LEFT$ and $\lobs\subseteq \RIGHT_2$.
  If $\Frame_1$ $f$-limits $\lsrc$-to-$\lcut$ flow, then $\Frame_2$
  $f$-limits $\lsrc$-to-$\lobs$ flow.
\end{repeatthm}
\begin{proof}
  By the hypothesis, $f$ is a blur operator.  Letting
  $\bnd_o\in\lrunstwo\lobs$, we want to show that
  $\cmpttwo{\lobs}{\lsrc}{\bnd_o}$ is an $f$-blurred set,
  i.e.~$\cmpttwo{\lobs}{\lsrc}{\bnd_o}=f(\cmpttwo{\lobs}{\lsrc}{\bnd_o})$.

  For convenience, let $S_c=\cmpttwo{\lobs}{\lcut}{\bnd_o}$.  By
  Lemma~\ref{lemma:cut:two:frames}, 
  $$ \cmpttwo{\lobs}{\lsrc}{\bnd_o} =
  \bigcup_{\bnd_c\in S_c} \cmptone{\lcut}{\lsrc}{\bnd_c} ;
  $$
  thus, we must show that the latter is $f$-blurred.  By the
  assumption that each $\cmptone{\lcut}{\lsrc}{\bnd_c}$ is
  $f$-blurred, we have
  $\cmptone{\lcut}{\lsrc}{\bnd_c}=f(\cmptone{\lcut}{\lsrc}{\bnd_c})$.
  Using this and the union property (Eqn.~\ref{eq:union}):
  \begin{eqnarray*}
    \bigcup_{\bnd_c\in S_c} \cmptone{\lcut}{\lsrc}{\bnd_c} 
    & = & \bigcup_{\bnd_c\in S_c} f(\cmptone{\lcut}{\lsrc}{\bnd_c}) \\
    & = & f(\bigcup_{\bnd_c\in S_c} \cmptone{\lcut}{\lsrc}{\bnd_c}) ,
  \end{eqnarray*}
  Hence, $\cmpttwo{\lobs}{\lsrc}{\bnd_o}$ is $f$-blurred.  \myqed
\end{proof}



\end{document}